\documentclass{article} 
\usepackage{iclr2027_conference,times}

\usepackage{graphicx}
\usepackage{fancyhdr}
\usepackage{hyperref}
\usepackage{url}
\usepackage{subcaption}
\usepackage{amsmath}
\usepackage{amsthm}
\usepackage{amsfonts}
\usepackage{algorithm}
\usepackage{algorithmic}
\usepackage{booktabs}

\usepackage{amsmath,amsfonts,bm}

\def\eqref#1{equation~\ref{#1}}

\def\1{\bm{1}}

\DeclareMathAlphabet{\mathsfit}{\encodingdefault}{\sfdefault}{m}{sl}
\SetMathAlphabet{\mathsfit}{bold}{\encodingdefault}{\sfdefault}{bx}{n}

\fancypagestyle{labpaperfirst}{
  \fancyhf{}
  \fancyhead[L]{%
    \includegraphics[width=4.2in]{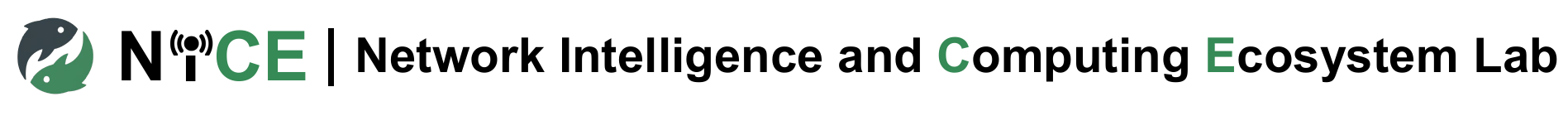}
  }
  \renewcommand{\headrulewidth}{0pt}
}

\fancypagestyle{labpaper}{
    \fancyhf{}
    \renewcommand{\headrulewidth}{0pt}
    \renewcommand{\footrulewidth}{0pt}
}

\newtheorem{assumption}{Assumption}
\newtheorem{proposition}{Proposition}

\theoremstyle{remark}
\newtheorem{remark}{Remark}

\title{
TwinJEPA: Action-Preferred Predictive Representations
for Goal-Conditioned Control
}

\author{
Feiran You$^{1}$, \quad Hongyang Du$^{1,*}$
\\[4pt]
The University of Hong Kong\\[4pt]
}

\iclrfinalcopy

\begin{document}

\maketitle
\thispagestyle{labpaperfirst}

\begingroup
\renewcommand{\thefootnote}{*}
\footnotetext{Corresponding author.}
\endgroup
\thispagestyle{labpaperfirst}

\pagestyle{labpaper}

\begin{abstract}

Joint-Embedding Predictive Architectures (JEPAs) have recently emerged as a
promising paradigm for representation learning by predicting future latent
states without reconstructing observations. Recent work has adapted JEPA-style latent prediction to offline zero-shot control through action-conditioned temporal prediction, enabling representations to capture long-horizon dynamics from fixed behavioral data. However, transition-level predictive objectives largely supervise actions in isolation and provide limited information about which actions are preferable when similar states admit different goal-conditioned outcomes. We introduce \textbf{TwinJEPA}, a framework for learning action-preferred predictive representations by augmenting JEPA-based control with offline-mined action-preference supervision. TwinJEPA identifies approximately matched states across offline trajectories and constructs preference pairs via goal-conditioned reward relabeling. It then learns two complementary objectives: reward-gap regression, which preserves the magnitude of outcome differences, and preference classification, which captures the relative ordering of alternative actions. Both objectives are used only during training and incur no additional inference-time cost. We evaluate TwinJEPA across long-horizon navigation and continuous-control
benchmarks, with both state-based and pixel-based observations. TwinJEPA yields positive benchmark-level mean differences across all matched state-based evaluations, while analyses across domains and observation modalities indicate that larger gains tend to arise when local action alternatives provide more informative outcome contrasts. The results suggest that local action-comparison supervision can complement temporal prediction by encouraging JEPA-based representations to retain both long-horizon temporal structure and fine-grained distinctions among locally observed actions for offline zero-shot control.

\end{abstract}
\section{Introduction}

Joint-Embedding Predictive Architectures (JEPAs) have emerged as a promising
paradigm for representation learning by predicting target representations
in latent space rather than reconstructing detailed
observations~\cite{assran2023self}.
By learning through latent prediction, JEPAs encourage representations to
retain structured and predictive information while suppressing observation
details less relevant to the prediction problem~\cite{chen2026got}.
This principle is particularly appealing for sequential decision-making,
where actions naturally condition how the system state evolves over time.
Action-conditioned latent prediction can therefore capture behavioral
dynamics and long-horizon temporal structure from sequential experience.
Recent methods such as TD-JEPA instantiate this idea for offline control
through action-conditioned temporal prediction~\cite{bagatella2025td}.

When applied to decision-making, however, temporal prediction alone does not
fully specify what information a representation should preserve about
alternative actions.
A transition-wise predictive objective learns how an observed action relates
to its subsequent outcome, but does not explicitly characterize how it
differs from plausible alternatives in a similar local state.
This distinction is important because control depends not only on modeling
how observed behaviors evolve over time, but also on preserving information
that separates actions leading to different downstream outcomes.
Two actions can originate from nearly identical states while producing
substantially different consequences, yet a standard temporal-prediction
objective supervises them independently.
As a result, a representation can capture long-horizon behavioral structure
while leaving fine-grained distinctions among locally available actions only
indirectly constrained~\cite{ding2025understanding}.

\begin{figure*}[t]
    \centering
    \includegraphics[width=\textwidth]{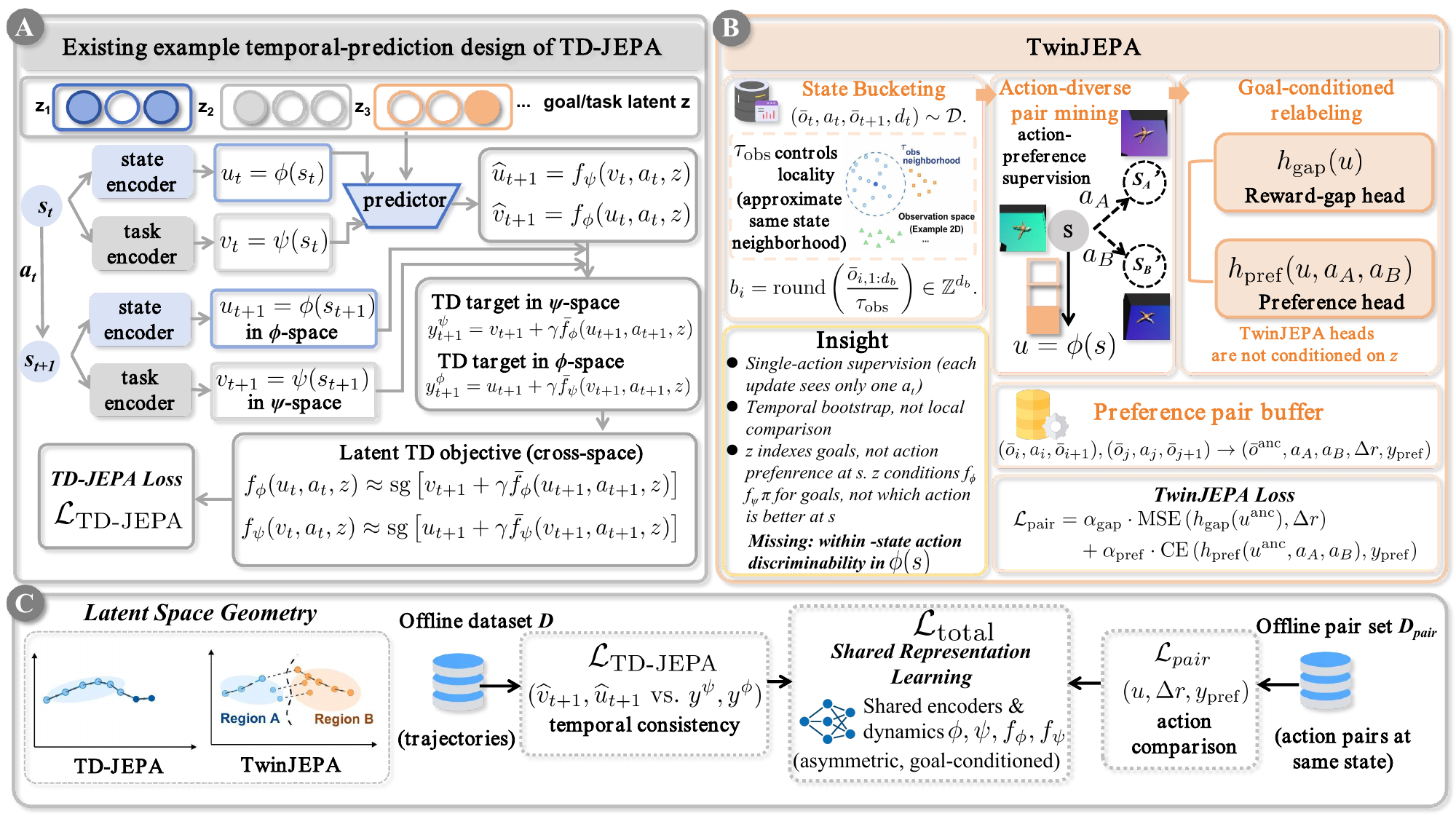}
    \caption{
    \textbf{Overview and key insight of TwinJEPA.}
    \textbf{A:} A transition-wise temporal-prediction objective
    such as TD-JEPA learns goal-conditioned temporal structure from
    action-conditioned transitions, but does not explicitly compare alternative
    actions in similar states.
    \textbf{B:} TwinJEPA mines action-diverse transition pairs from local state
    neighborhoods and derives reward-gap and preference targets through
    goal-conditioned relabeling.
    \textbf{C:} Temporal prediction and pairwise supervision jointly train the
    shared representation to preserve trajectory-level temporal structure while
    encoding local distinctions between actions associated with different
    goal-conditioned outcomes.
    The auxiliary heads are used only during training and introduce no
    additional inference-time computation.
}
    \label{fig:large-figure}
\end{figure*}

To address this gap, we propose \textbf{TwinJEPA}, a JEPA-based framework
for learning action-preferred predictive representations by augmenting
temporal prediction with local action-comparison supervision mined from
offline behavioral data.
TwinJEPA strengthens fine-grained distinctions among alternative actions in
similar local states while preserving the long-horizon structure captured by
the predictive backbone.
It groups transitions into approximate state neighborhoods and retains pairs
with sufficiently different actions.
Goal-conditioned reward relabeling then identifies an outcome contrast for
each pair, yielding two complementary supervisory signals:
reward-gap regression captures the magnitude of the outcome difference,
whereas preference classification captures its ordering.
These signals are applied through lightweight auxiliary heads, while the
temporal-prediction backbone, policy-learning procedure, and task-inference
mechanism remain unchanged.
TwinJEPA therefore introduces comparative information into the learned
representation without human preference labels, online interaction, or
inference-time computation, as illustrated in
Figure~\ref{fig:large-figure}.

TwinJEPA is complementary to several related lines of work.
JEPA-based control methods such as TD-JEPA primarily supervise
representations through transition-wise prediction along observed
trajectories~\cite{bagatella2025td}, whereas TwinJEPA adds explicit
comparisons between alternative actions in locally similar states.
Offline and zero-shot representation-learning methods likewise seek reusable
control representations, but do not necessarily impose such local
action-comparison structure~\cite{rao2026skyjepa,yang2026reward}.
TwinJEPA also differs from preference-based RL: rather than learning a reward
model or directly optimizing a policy from preference feedback, it derives
pairwise targets automatically from offline data and uses them solely as
auxiliary representation-learning supervision.
Additional discussion is provided in Appendix~\ref{sec:related-work}.

We evaluate TwinJEPA under controlled matched protocols across long-horizon
navigation and continuous-control benchmarks, including OGBench AntMaze
~\cite{park2025ogbench}, D4RL Maze2D, and DMC-based Walker, Cheetah, and
Quadruped~\cite{tassa2018deepmind}.
Matched-seed analyses assess effect consistency and uncertainty.
Across state-based evaluations, TwinJEPA yields positive mean differences on
all benchmarks, with substantial variation in magnitude and seed-level
consistency.
Task-level analyses show that these gains are structured rather than uniform:
AntMaze and Maze2D exhibit goal-specific changes, Walker and Cheetah improve
more on selected control objectives, and Quadruped shows positive mean
differences across all four tasks and five matched seeds.
We also evaluate the same action-preference objective with pixel observations
on ExORL to assess its applicability across modalities.
Component ablations show that reward-gap regression and preference
classification contribute differently across domains, with no component or
combination uniformly dominating across settings.
Our contributions are summarized as follows:
\begin{itemize}

\item
We introduce \textbf{TwinJEPA}, a framework for learning action-preferred predictive representations by augmenting JEPA-based
temporal supervision with explicit local action comparisons.
The resulting objective preserves long-horizon temporal structure while
strengthening fine-grained distinctions among alternative actions in locally
similar states, without changing the downstream control or inference
procedure.

\item
We develop a fully offline approach for deriving comparative supervision from
existing behavioral data.
TwinJEPA mines action-diverse transitions from local state neighborhoods and
uses goal-conditioned relabeling to construct two distinct targets,
reward-gap magnitude and action-preference ordering, without human preference
annotations or additional environment interaction.
We further characterize these objectives theoretically through reward-gap
decodability and local action ranking on the mined comparison distribution.

\item
We provide controlled matched-seed evaluations across navigation and
continuous-control benchmarks with both state-based and pixel-based
observations.
TwinJEPA yields positive mean differences across all matched state-based
evaluations, while task-level, seed-level, and component analyses reveal
substantial cross-domain structure in the resulting effects and in the
contributions of the two comparative signals.

\end{itemize}
\section{Preliminaries}
\label{sec:preliminaries}

\subsection{Offline Goal-Conditioned and Zero-Shot Control}
\label{sec:offline-gc}

We consider offline goal-conditioned control from a fixed dataset
$
\mathcal{D}
=
\left\{
(o_t,a_t,o_{t+1},\delta_t)
\right\}_{t=1}^{|\mathcal D|},
$
where $o_t$ and $a_t$ denote the observation and action, respectively, and
$\delta_t\in\{0,1\}$ indicates episode termination.
No additional environment interaction is available during representation or
policy learning.
The goal is to learn reusable behaviors conditioned on a target task at
deployment~\cite{bortkiewicz2025accelerating}. Let $h(o_t)$ denote the learned state representation and $z$ a latent
goal or task variable, with goal-conditioned policy
$
\pi_\eta(a_t \mid h(o_t),z).
$
In the zero-shot setting, the downstream task is unspecified during
representation learning. Task latents are instead sampled during offline
training,
$
z\sim p(z),
a_t\sim
\pi_\eta(\cdot\mid h(o_t),z),
$
yielding a family of latent-conditioned behaviors~\cite{sobal2026learning}.
The corresponding control objective is
\begin{equation}
\label{eq:zero-shot-objective}
\eta^\star
=
\arg\max_{\eta}
\mathbb{E}_{z\sim p(z)}
\left[
\sum_{t=0}^{T}
\gamma^t r_z(o_t,a_t)
\right],
\end{equation}
where $r_z$ denotes the reward associated with task latent $z$.
At deployment, a task representation $z^\star$ is inferred from
task-specific information and supplied to the pretrained policy,
$
a_t
\sim
\pi_{\eta^\star}
\left(
\cdot\mid h(o_t),z^\star
\right),
$
without additional policy optimization or online fine-tuning. TwinJEPA focuses on learning $h$ while retaining the policy and zero-shot
task-inference mechanisms of the underlying temporal-prediction framework
and modifying only the supervision used for representation learning.

\subsection{JEPA-Based Temporal Predictive Control}
\label{sec:temporal-predictive}

JEPAs learn representations by predicting future embeddings rather than
reconstructing observations
~\cite{lecun2022path,klindt2026does,destrade2025value}.
This latent-prediction principle is attractive for control because it
emphasizes predictive structure without reconstructing every pixel- or
state-level detail
~\cite{wang2026embodied,chen2025vl,assran2023self}.
TD-JEPA applies this principle to offline goal-conditioned control through
asymmetric cross-space temporal prediction~\cite{bagatella2025td}.
Two encoders produce
$
u_t=\phi_{\theta_\phi}(o_t),
\qquad
v_t=\psi_{\theta_\psi}(o_t),
$
and action- and task-conditioned predictors estimate future representations
in the opposite latent spaces,
$
\widehat v_{t+1}
=
f_\phi(u_t,a_t,z),
\widehat u_{t+1}
=
f_\psi(v_t,a_t,z).
$
Momentum networks provide bootstrapped targets
$y_t^\phi$ and $y_t^\psi$, yielding
\begin{equation}
\mathcal L_{\rm TD}
=
\left\|
\widehat v_{t+1}
-
\operatorname{sg}(y_t^\phi)
\right\|_2^2
+
\left\|
\widehat u_{t+1}
-
\operatorname{sg}(y_t^\psi)
\right\|_2^2,
\label{eq:tdjepa-loss}
\end{equation}
where $\operatorname{sg}(\cdot)$ denotes stop-gradient. TD-JEPA also applies
representation regularization $\mathcal L_{\rm orth}$.
Since TwinJEPA leaves both terms unchanged, we define
$
\mathcal L_{\rm base}
=
\mathcal L_{\rm TD}
+
\mathcal L_{\rm orth}
$
as the representation-learning backbone.
Regularizer details are deferred to the appendix. Importantly, $\mathcal L_{\rm TD}$ is transition-wise.
For two transitions
$(o_i,a_i,o_i')$ and $(o_j,a_j,o_j')$ with locally similar observations but
different actions, the objective trains each prediction against its future
target, without directly comparing the goal-conditioned outcomes of
$a_i$ and $a_j$.
Thus, temporal prediction constrains each observed action individually while
leaving the relative structure among local action alternatives only
indirectly specified.
TwinJEPA adds this missing comparative signal without replacing the
temporal-prediction backbone.
Section~\ref{sec:method} describes how local action pairs are mined and how
reward-gap and preference targets are added to the representation objective.
\section{Method}
\label{sec:method}

TwinJEPA augments the temporal-prediction backbone in
Section~\ref{sec:temporal-predictive} with offline-mined local
action-comparison supervision, while leaving the temporal objective,
representation regularization, policy, and task-inference procedure unchanged.
It consists of (i) mining and goal-conditioned relabeling of action-diverse
transition pairs and (ii) lightweight auxiliary heads for comparative
supervision.
Both components are used only during training and add no inference-time cost.

\subsection{Offline Local Action-Comparison Mining}
\label{sec:cf-mining}

To complement the transition-wise objective in Eq.~\ref{eq:tdjepa-loss},
TwinJEPA constructs a fixed pair buffer
$\mathcal{D}_{\mathrm{pair}}$ containing different actions observed in locally
similar states.
The buffer is constructed once from the offline dataset before representation
learning.
\paragraph{Local state neighborhoods.}
Exact state revisits are rare in continuous-control datasets, so we approximate
local state similarity using a fixed matching representation
$q(\bar{o})$.
For the state-based experiments in this paper,
$q(\bar{o})=\bar{o}_{1:d_b}$ selects a task-relevant subset of the normalized
observation; for example, $d_b=2$ corresponds to the XY position in AntMaze,
whereas the full proprioceptive state can be used in continuous-control
domains.
Given a neighborhood resolution $\tau_{\mathrm{obs}}>0$, transition $i$ is
assigned to
\begin{equation}
\label{eq:method-bucket}
b_i
=
\operatorname{round}
\left(
\frac{q(\bar{o}_i)}{\tau_{\mathrm{obs}}}
\right).
\end{equation}
Transitions sharing a bucket index are treated as approximate local matches.
The locality parameter $\tau_{\mathrm{obs}}$ controls the trade-off between
state similarity and pair availability: smaller values yield tighter
neighborhoods but fewer comparisons, while larger values increase pair
density at the cost of greater state variation.
We therefore treat the matching representation and neighborhood resolution
as pair-construction hyperparameters; benchmark-specific settings and
$\tau_{\mathrm{obs}}$ sensitivity results are reported in
Appendix~\ref{app:tol-sensitivity}. Within each neighborhood, we retain action-diverse pairs satisfying
$
\|a_i-a_j\|_2 \geq \epsilon_a,
$
where $\epsilon_a$ filters out nearly identical actions.
The resulting pairs are locally matched in state while remaining sufficiently
distinct in action.

\paragraph{Goal-conditioned comparative targets.}
Consider a retained local pair $(i,j)$ with
$a_A=a_i$ and $a_B=a_j$.
We construct comparative targets using the same offline goal-relabeling
mechanism available during representation learning.
Specifically, candidate goals are sampled from states contained in the fixed
offline dataset (e.g., through future-state relabeling) and do not include
held-out evaluation goals.
Thus, the pair construction uses only training-time information and does not
introduce downstream-task or evaluation-goal leakage. For each candidate relabeled goal $k$, let
$r_A^{(k)}$ and $r_B^{(k)}$ denote the goal-conditioned rewards associated
with the corresponding next states.
We select the relabeled goal that maximizes the reward contrast between the
two locally matched transitions,
\begin{equation}
\label{eq:method-k-star}
k^\star
=
\arg\max_k
\left|
r_A^{(k)}-r_B^{(k)}
\right|.
\end{equation}
Because all candidate goals use the same goal-conditioned reward definition
and numerical scale, this maximization compares outcome contrasts on a common
scale rather than favoring goals with intrinsically larger reward ranges.
It therefore selects the most action-discriminative supervision available
from the local pair. Using $k^\star$, we define the reward-gap magnitude
$
\Delta r
=
\left|
r_A^{(k^\star)}
-
r_B^{(k^\star)}
\right|
$
and the preferred-action label
$
y_{\mathrm{pref}}
=
\mathbb{I}
\left[
r_B^{(k^\star)}
>
r_A^{(k^\star)}
\right],
$
where $y_{\mathrm{pref}}=0$ favors $a_A$ and
$y_{\mathrm{pref}}=1$ favors $a_B$.
Pairs with negligible reward contrast,
$\Delta r\leq 10^{-6}$, are removed.
Both targets are obtained entirely from goal-relabeled offline transitions,
requiring neither human preference annotations, evaluation-task information,
nor additional environment interaction.

\paragraph{Local anchor construction.}
For each matched pair, we use the first transition as a deterministic anchor,
$\bar{o}^{\mathrm{anc}}=\bar{o}_i$, and store
$
\left(
\bar{o}^{\mathrm{anc}},
a_A,
a_B,
\Delta r,
y_{\mathrm{pref}}
\right)
\in \mathcal{D}_{\mathrm{pair}}.
$
As the two transitions satisfy the local matching criterion,
$\bar{o}_i$ and $\bar{o}_j$ provide nearby state contexts for comparing
$a_A$ and $a_B$.
Using a single anchor yields a consistent representation of
each local action comparison.
Pair construction continues until the budget $N_{\mathrm{pair}}$ is reached,
after which $\mathcal{D}_{\mathrm{pair}}$ is fixed for representation
learning.

\paragraph{Extension to pixel observations.}
The same comparative supervision extends to pixel-based ExORL experiments.
We construct pairs from cached visual observations using the
modality-specific procedure described in Appendix~\ref{app:rgb-obs}, while retaining the
same reward-gap and preference targets.
This separates the comparative learning objective from the particular
matching representation used to identify local transition pairs.

\subsection{Pairwise Comparative Supervision}
\label{sec:cp-supervision}

Given a sample
$
\left(
\bar{o}^{\mathrm{anc}},
a_A,
a_B,
\Delta r,
y_{\mathrm{pref}}
\right)
\sim
\mathcal{D}_{\mathrm{pair}},
$
the anchor observation is encoded using the same state encoder as the
temporal-prediction backbone,
$
u^{\mathrm{anc}}
=
\phi_{\theta_\phi}
\left(
\bar{o}^{\mathrm{anc}}
\right).
$
TwinJEPA applies two lightweight auxiliary heads to the shared
representation. The reward-gap head is conditioned only on the anchor representation and
predicts the magnitude of the local outcome contrast,
$
\widehat{\Delta r}
=
h_{\mathrm{gap}}
\left(
u^{\mathrm{anc}}
\right).
$
This state-only objective encourages the representation to retain information
predictive of outcome-contrast magnitude under the mined comparison
distribution, rather than directly modeling a particular action pair. In contrast, the preference head explicitly conditions on both candidate
actions together with the anchor representation,
$
\widehat p
=
h_{\mathrm{pref}}
\left(
u^{\mathrm{anc}},
a_A,
a_B
\right),
$
where $\widehat p\in(0,1)$ denotes the predicted probability that $a_B$ is
preferred over $a_A$.
The resulting pairwise objective is
\begin{equation}
\label{eq:cp-objective}
\mathcal{L}_{\mathrm{pair}}
=
\alpha_{\mathrm{gap}}
\left(
\widehat{\Delta r}
-
\Delta r
\right)^2
+
\alpha_{\mathrm{pref}}
\operatorname{BCE}
\left(
\widehat p,
y_{\mathrm{pref}}
\right).
\end{equation}

The two terms capture complementary aspects of the same mined local
comparison.
Reward-gap regression provides magnitude-sensitive supervision by encouraging
the representation to preserve information predictive of local outcome
contrasts, whereas preference classification provides pair-specific ordering
supervision over $(a_A,a_B)$.
We do not assume that either signal, or their combination, dominates across
all environments; their relative utility is evaluated in
Section~\ref{sec:component-analysis}. The task latent $z$ is not supplied to either auxiliary head.
Goal-dependent information enters only through the relabeled targets
$\Delta r$ and $y_{\mathrm{pref}}$.
This prevents trivial conditioning on the task representation: the reward-gap
objective promotes state information predictive of local contrast magnitude,
while the preference objective combines this state context with the candidate
actions to learn their relative ordering.
The original temporal predictors and policy remain conditioned on $z$ as in
TD-JEPA.
\subsection{Joint Representation Learning}
\label{sec:training}

Let
$
\mathcal{L}_{\mathrm{base}}
=
\mathcal{L}_{\mathrm{TD}}
+
\mathcal{L}_{\mathrm{orth}}
$
denote the unchanged TD-JEPA representation objective introduced in
Section~\ref{sec:temporal-predictive}.
TwinJEPA optimizes
$
\mathcal{L}_{\mathrm{repr}}
=
\mathcal{L}_{\mathrm{base}}
+
\lambda
\mathcal{L}_{\mathrm{pair}},
$
where $\lambda$ controls the contribution of comparative supervision.
Setting $\lambda=0$ exactly recovers the matched temporal-prediction baseline
under the same architecture, offline data, optimizer, policy learner, and
training budget. At each training iteration, a transition minibatch is independently sampled
from $\mathcal{D}$ and a comparison minibatch from
$\mathcal{D}_{\mathrm{pair}}$.
The encoders and temporal predictors are updated through
$\mathcal{L}_{\mathrm{base}}$, while the shared encoder additionally receives
gradients from $\mathcal{L}_{\mathrm{pair}}$.
The auxiliary heads are optimized only through their corresponding pairwise
targets.
The Flow Behavior Cloning (FlowBC) actor
~\cite{park2026horizon} is trained using the same policy-learning procedure
as the baseline, with the state representation detached during the actor
update.
Target encoders and predictors are updated using the same exponential moving
average procedure as TD-JEPA. 
The complete state-based training procedure is provided in Appendix~\ref{detail} (Algorithm~\ref{alg:twinjepa}).

\subsection{Inference and Computational Overhead}
\label{sec:deployment}

TwinJEPA changes only the representation-learning objective; downstream
control and test-time inference remain identical to the matched TD-JEPA
baseline.
The learned representation is used by the FlowBC actor as
$
a_t
=
\pi_\eta
\left(
\phi_{\theta_\phi}(\bar{o}_t),
z
\right).
$
For held-out goals, the task latent $z$ is inferred using the same procedure
as in TD-JEPA, enabling zero-shot control without policy optimization or
online fine-tuning.
The mined comparison pairs and auxiliary heads are required only during
offline training and are not used at inference time.
Consequently, TwinJEPA incurs additional offline pair-construction and
training cost, while preserving the baseline inference-time computational
cost.

\section{Experimental Results}
\label{sec:results}

\subsection{Experimental Setup and Statistical Protocol}
\label{sec:exp-setup}

\paragraph{Benchmarks.}
We evaluate TwinJEPA on OGBench AntMaze, D4RL Maze2D, and DMC-based
Walker, Cheetah, and Quadruped, covering navigation and continuous-control
settings with diverse reward structures, embodiments, and task geometries.
All methods are trained from fixed offline replay buffers.
We use TD-JEPA as the matched baseline, reproduced from the authors'
implementation~\cite{bagatella2025td}.
Within each matched comparison, TD-JEPA and TwinJEPA share the architecture,
offline data, optimization schedule, training budget, random seeds,
task-inference procedure, and evaluation protocol, with TwinJEPA adding only
the offline-mined action-comparison objective.
Table~\ref{tab:published-benchmark-context} additionally reports published
results for Laplacian, ICVF, HILP, FB, RLDP, BYOL, and BYOL-$\gamma$ where
available, providing broader comparison with existing methods.

\paragraph{Evaluation and statistical protocol.}
We use benchmark-specific protocols aligned with TD-JEPA evaluations and
retain a shared auxiliary configuration across domains.
AntMaze and Cheetah use ten matched seeds, while Maze2D, Walker, and
Quadruped use five; training budgets, task-inference settings, local matching
parameters, and sensitivity analyses are reported in the appendix.
For matched seeds, we analyze paired differences
$
d_i = R_i^{\mathrm{Twin}} - R_i^{\mathrm{TD}},
$
where $R_i$ denotes either success rate or mean task return.
We report the mean difference $\bar d$, percentile $95\%$ confidence intervals
from $20{,}000$ paired bootstrap resamples, exact two-sided sign-flip
permutation $p$-values, and the number of seeds with $d_i>0$.
For small $n$, the exact permutation distribution is discrete, so confidence
interval exclusion and $p<0.05$ need not coincide.

\begin{table*}[t]
\centering
\caption{
Matched-seed comparison of TwinJEPA and TD-JEPA.
The 95\% CIs are based on $20{,}000$ paired bootstrap resamples, and
$p_{\mathrm{2s}}$ denotes the exact two-sided sign-flip permutation
$p$-value.
Succ. (pp) denotes success rate in percentage points; returns are reported
on their native scales.
}
\label{tab:primary-paired-results}

{\small
\setlength{\tabcolsep}{3.2pt}
\renewcommand{\arraystretch}{1.06}

\resizebox{\textwidth}{!}{%
\begin{tabular}{@{}llccccccc@{}}
\hline
Setting
& Metric
& $n$
& TD-JEPA
& TwinJEPA
& $\bar d$
& 95\% CI
& $p_{\mathrm{2s}}$
& $d_i>0$ \\
\hline

AntMaze medium-navigate
& Succ. (pp) & 10 & 68.8 & 73.2 & $+4.4$
& $[-7.4,+16.4]$ & $0.518$ & $5/10$ \\

AntMaze large-navigate
& Succ. (pp) & 10 & 42.2 & 43.0 & $+0.8$
& $[-7.0,+8.4]$ & $0.883$ & $5/10$ \\

AntMaze medium-stitch
& Succ. (pp) & 10 & 51.4 & 52.0 & $+0.6$
& $[-13.6,+13.6]$ & $0.959$ & $6/10$ \\

AntMaze large-stitch
& Succ. (pp) & 10 & 37.8 & 41.2 & $+3.4$
& $[-3.0,+8.6]$ & $0.336$ & $7/10$ \\
\hline

AntMaze-4 seed mean
& Succ. (pp) & 10 & 50.1 & 52.4 & $+2.3$
& $[-2.8,+7.8]$ & $0.447$ & $5/10$ \\

Maze2D-medium
& Succ. (pp) & 5 & 9.2 & 14.4 & $+5.2$
& $[+1.6,+10.0]$ & $0.125$ & $4/5$ \\

DMC Walker, Fixed-HP
& Return & 5 & 774.3 & 782.7 & $+8.4$
& $[-9.9,+26.6]$ & $0.438$ & $3/5$ \\

Cheetah, state
& Return & 10 & 638.0 & 671.7 & $+33.7$
& $[-1.0,+68.8]$ & $0.104$ & $7/10$ \\

Quadruped, $\tau_{\mathrm{obs}}=1.0$
& Return & 5 & 559.4 & 599.8 & $+40.4$
& $[+20.6,+60.3]$ & $0.0625$ & $5/5$ \\
\hline
\end{tabular}%
}
}

\end{table*}

\begin{figure*}[t]
\centering
\includegraphics[width=\linewidth]
{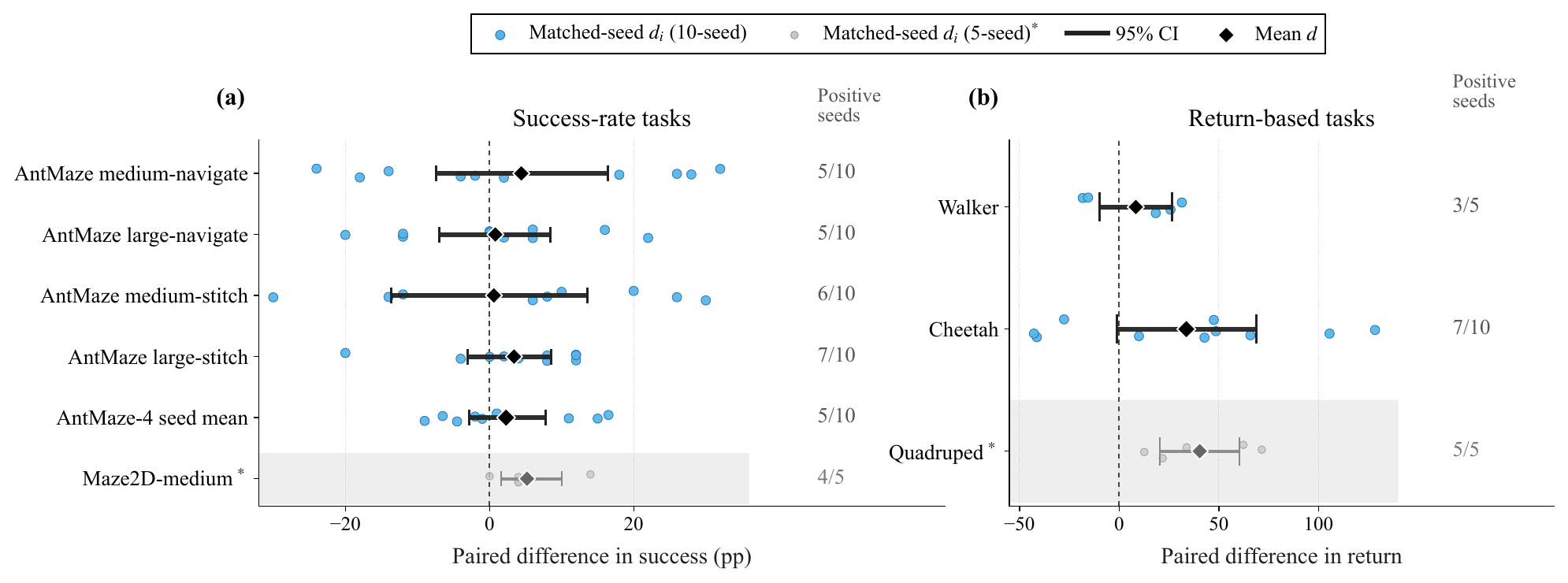}
\caption{
Matched-seed effects across paper-aligned evaluations.
\textbf{(a)} Success-rate differences for the four AntMaze domains,
the AntMaze-4 aggregate, and Maze2D-medium.
\textbf{(b)} Return differences for Walker, Cheetah, and Quadruped.
Small markers show paired differences, diamonds their means, and horizontal
bars paired-bootstrap $95\%$ confidence intervals.
The dashed line marks zero; right-hand annotations give the number of
positive seeds.
}
\label{fig:paired-effect-forest}
\end{figure*}

\paragraph{Evaluation questions.}
Our experiments address three questions.
\textbf{RQ1:} Does TwinJEPA improve the matched TD-JEPA baseline under
controlled state-based evaluations across navigation and continuous-control
domains?
\textbf{RQ2:} How are the resulting effects distributed across domains,
downstream tasks, and random seeds?
\textbf{RQ3:} How do reward-gap regression and preference classification
contribute to these effects across domains?
\subsection{Matched State-Based Evaluation (\textbf{RQ1})}
\label{sec:matched-main}

To address RQ1, we compare TwinJEPA with the matched TD-JEPA baseline under
paper-aligned protocols across AntMaze, Maze2D, Walker, Cheetah, and
Quadruped.
Within each comparison, the two methods share the architecture, offline data,
optimization schedule, training budget, random seeds, task-inference
procedure, and evaluation protocol, differing only in the additional
action-comparison objective. TwinJEPA yields positive mean differences in every state-based evaluation
reported in Table~\ref{tab:primary-paired-results}.
On AntMaze, all four domain-level means increase, producing an AntMaze-4
change from $50.1\%$ to $52.4\%$ ($\bar d=+2.3$ pp).
Maze2D increases from $9.2\%$ to $14.4\%$ ($\bar d=+5.2$ pp).
Among the continuous-control domains, Walker increases from $774.3$ to
$782.7$ ($\bar d=+8.4$), Cheetah from $638.0$ to $671.7$
($\bar d=+33.7$), and Quadruped from $559.4$ to $599.8$
($\bar d=+40.4$).

The magnitude and uncertainty of these effects vary across benchmarks.
Quadruped shows the strongest matched result, with all five paired
differences positive and a paired-bootstrap $95\%$ confidence interval of
$[+20.6,+60.3]$.
Cheetah also exhibits a substantial positive mean difference, while the
AntMaze and Walker effects are smaller and more variable across seeds.
These results establish a consistent positive direction at the level of
benchmark means under the shared matched protocol; RQ2 examines how these
effects decompose across individual tasks and seeds.
\begin{table}[t]
\centering
\caption{
Published benchmark context on evaluations overlapping with TwinJEPA.
Panel~A reports zero-shot success (\%) on the four OGBench AntMaze domains.
Panel~B reports state- and pixel-based returns on Walker, Cheetah, and
Quadruped.
Published rows provide cross-paper context; the final two rows use our
paper-aligned protocols.
}
\label{tab:published-benchmark-context}
\begin{minipage}[t]{0.445\linewidth}
\vspace{0pt}
\centering
{\footnotesize
\setlength{\tabcolsep}{1.35pt}
\renewcommand{\arraystretch}{1.02}

\begin{tabular}{@{}lccccc@{}}
\toprule
\multicolumn{6}{c}{\textit{A. OGBench AntMaze}} \\
\midrule
Method & mn & ln & ms & ls & Avg. \\
\midrule

Laplacian
& 50.0 & 21.6 & 21.4 & 11.8 & 26.2 \\

ICVF$^{\ast}$
& 79.8 & 58.4 & 39.0 & 13.2 & 47.6 \\

HILP
& 83.6 & 52.6 & 50.6 & 12.2 & 49.8 \\

FB
& 73.0 & 36.8 & 70.4 & 49.8 & 57.5 \\

RLDP
& 74.6 & 36.4 & 58.4 & 19.6 & 47.3 \\

BYOL$^{\ast}$
& 58.4 & 26.6 & 60.6 & 25.8 & 42.9 \\

BYOL-$\gamma^{\ast}$
& 51.4 & 21.8 & 45.6 & 20.2 & 34.8 \\

TD-JEPA$_{\mathrm{pub}}$
& 70.4 & 57.2 & 61.6 & 40.6 & 57.4 \\

\midrule

TD-JEPA$_{\mathrm{match}}$
& 68.8 & 42.2 & 51.4 & 37.8 & 50.1 \\

TwinJEPA
& 73.2 & 43.0 & 52.0 & 41.2 & 52.4 \\

\bottomrule
\end{tabular}
}
\end{minipage}
\hfill
\begin{minipage}[t]{0.545\linewidth}
\vspace{0pt}
\centering
{\footnotesize
\setlength{\tabcolsep}{1.15pt}
\renewcommand{\arraystretch}{1.02}

\begin{tabular}{@{}lcccccc@{}}
\toprule
\multicolumn{7}{c}{\textit{B. ExORL / DMC}} \\
\midrule

&
\multicolumn{2}{c}{Walker}
&
\multicolumn{2}{c}{Cheetah}
&
\multicolumn{2}{c}{Quadruped}
\\
\cmidrule(lr){2-3}
\cmidrule(lr){4-5}
\cmidrule(l){6-7}

Method
& State & RGB
& State & RGB
& State & RGB
\\
\midrule

HILP
& -- & 422.8
& -- & 333.0
& -- & 513.9 \\

FB
& -- & 324.4
& -- & 622.4
& -- & 475.4 \\

RLDP
& -- & 576.1
& -- & 605.3
& -- & 551.1 \\

BYOL-$\gamma^{\ast}$
& -- & 648.3
& -- & 679.8
& -- & 570.0 \\

TD-JEPA$_{\mathrm{pub}}$
& 785.2 & 738.9
& 688.7 & 706.0
& 691.4 & 626.7 \\

\midrule

TD-JEPA$_{\mathrm{match}}$
& 774.3 & 716.6
& 638.0 & 694.8
& 559.4 & 645.0 \\

TwinJEPA
& 782.7 & 741.8
& 671.7 & 686.4
& 599.8 & 650.7 \\

\bottomrule
\end{tabular}
}
\end{minipage}

\end{table}


\begin{figure}[t]
\centering
\includegraphics[
    width=0.9\linewidth
]{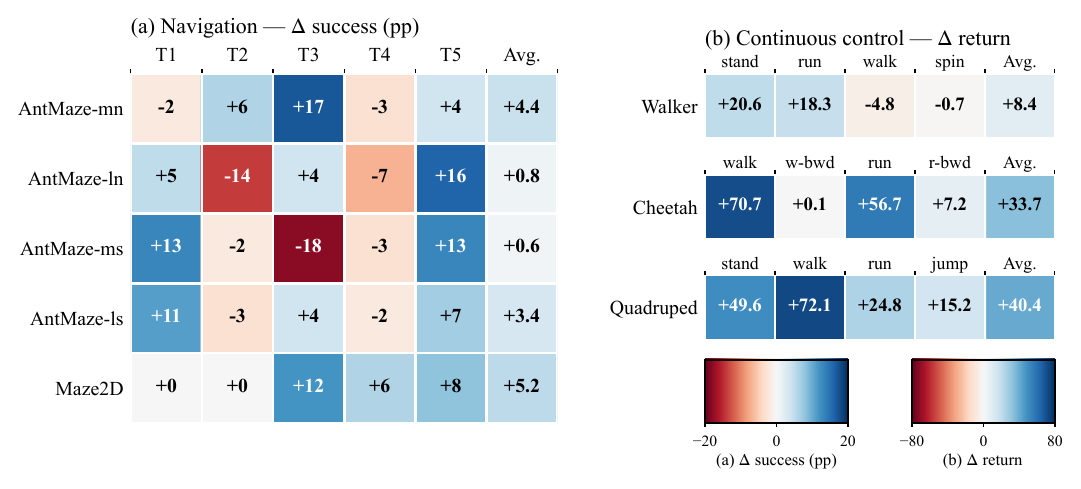}
\caption{
Task-level matched differences across navigation and continuous-control
benchmarks.
Panel~(a) reports TwinJEPA--TD-JEPA zero-shot success differences
(percentage points) for AntMaze and Maze2D.
Panel~(b) reports return differences for Walker, Cheetah, and Quadruped.
The final column gives the corresponding domain mean.
}
\label{fig:cross-domain-task-heatmap}
\end{figure}

\subsection{Effect Structure Across Domains and Tasks (\textbf{RQ2})}
\label{sec:cross-domain-results}

RQ1 establishes positive benchmark-level mean differences under the matched
state-based protocols.
RQ2 examines how these effects vary across tasks and random seeds.
Figure~\ref{fig:cross-domain-task-heatmap} shows that the gains are generally
structured rather than uniform: improvements are often concentrated on a
subset of tasks, while others change little or decrease. The navigation domains exhibit strongly goal-dependent effects.
Across AntMaze, positive and negative task-level differences coexist, with
the largest changes concentrated on particular goals.
Maze2D shows a similar pattern: its aggregate gain comes entirely from goals
on which the baseline already achieves nonzero success, while the remaining
goals are unchanged.
This suggests that local action-comparison supervision is most effective when
the mined comparisons capture distinctions relevant to the evaluated goal. Continuous-control domains show a related but distinct structure.
Walker and Cheetah concentrate their larger gains on selected locomotion
objectives, while other tasks change little.
Figure~\ref{fig:continuous-task-structure} shows that these larger effects
often occur where the matched baseline retains substantial performance
headroom.
Quadruped, however, improves across all four tasks despite substantial
variation in baseline return, indicating that headroom alone does not explain
the observed effects.
The benefit of comparative supervision therefore appears to depend on how
the mined comparison distribution aligns with the downstream objective. Seed-level results show a similar contrast
(Figure~\ref{fig:paired-effect-forest}).

AntMaze and Walker exhibit substantial run-to-run variation, Cheetah shows a
more consistently positive tendency, Maze2D is positive on most matched
seeds, and Quadruped is positive on all five matched runs.
Thus, the observed effects range from localized task-specific gains to broader
patterns that persist across both tasks and seeds. Overall, the effect of local action-comparison supervision is
task- and data-dependent rather than a uniform domain-wide shift.
Reward structure, offline coverage, local action diversity, and the
informativeness of the mined comparisons may all contribute to this
heterogeneity.
This is consistent with TwinJEPA as an auxiliary representation-learning
objective: downstream gains depend on whether the distinctions encoded by
the mined comparison distribution are useful for the evaluated task.
\begin{figure}[t]
\centering
\includegraphics[
    width=0.98\linewidth
]{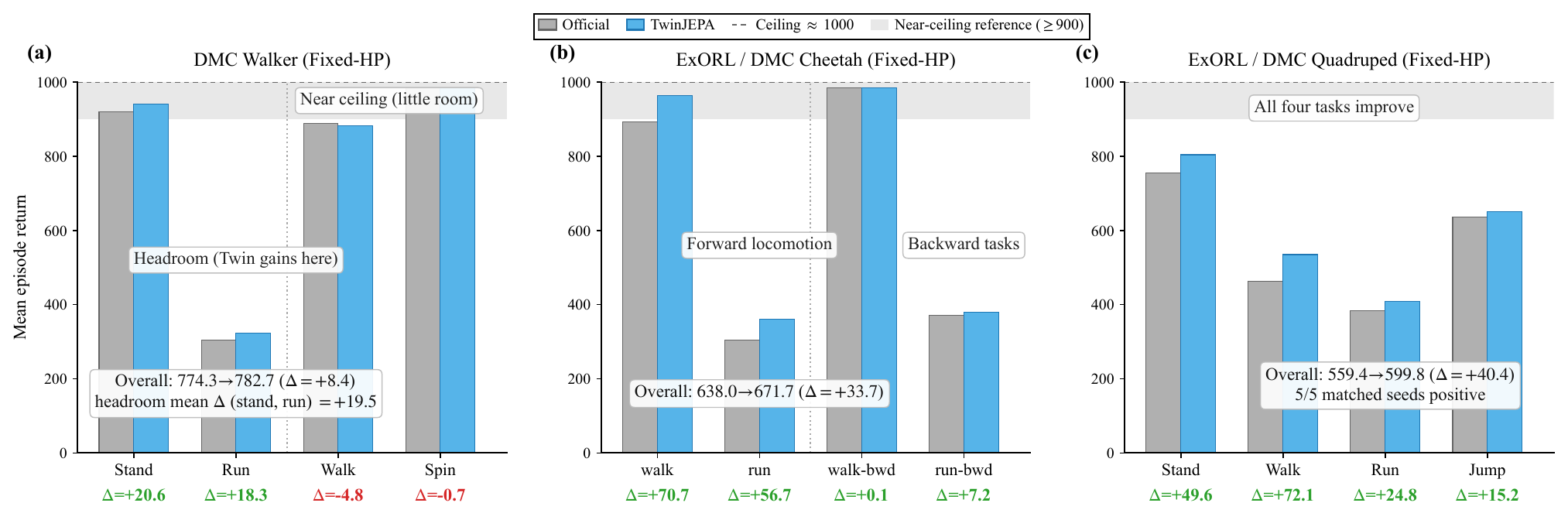}
\caption{
Task-level absolute returns on the three continuous-control domains.
Bars compare Official TD-JEPA and TwinJEPA under the paper-aligned
Fixed-HP protocols, with annotations showing
$\Delta=\mathrm{TwinJEPA}-\mathrm{TD\mbox{-}JEPA}$.
The shaded region marks near-ceiling performance.
}
\label{fig:continuous-task-structure}
\end{figure}

\subsection{Component Contributions Across Domains (\textbf{RQ3})}
\label{sec:component-analysis}

To address RQ3, we ablate reward-gap regression and preference classification
to examine how the two comparative signals contribute across domains.
Reward-gap regression supervises the magnitude of the selected outcome
contrast, whereas preference classification supervises its ordering.
Their empirical contributions are not interchangeable and vary across
evaluation settings. On AntMaze, the component ordering differs across domains.
Preference-only reaches $77.0\%$ on medium-navigate, exceeding the
$73.2\%$ full-objective result, whereas gap-only reaches $57.4\%$ on
medium-stitch, compared with $52.0\%$ for the full objective.
Thus, neither component is consistently dominant across the two AntMaze
settings. Cheetah exhibits a different pattern.
On the matched five-seed ablation subset, the full objective reaches $671.4$,
while gap-only and preference-only reach $654.5$ and $629.8$, respectively.
Here, gap-only retains more of the full-objective improvement, while the
combined objective achieves the highest mean return among the evaluated
variants.

Quadruped shows the largest separation between the full and isolated
objectives.
Full TwinJEPA reaches $599.8$, whereas gap-only and preference-only obtain
$554.4$ and $548.4$, respectively, both below the matched TD-JEPA baseline
of $559.4$.
In this setting, neither isolated objective retains the improvement observed
with the full configuration. The ablations therefore reveal no fixed ordering between reward-gap and
preference supervision across domains.
Their combination is strongest in Quadruped, while isolated components can
be more competitive on AntMaze and Cheetah.
We retain
$(\lambda,\alpha_{\mathrm{gap}},\alpha_{\mathrm{pref}})
=(0.1,1.0,0.5)$
as a shared cross-domain configuration rather than selecting
domain-specific component weights.
Additional component, $\lambda$-sensitivity, uncertainty, and seed-level
results are reported in
Appendix~\ref{app:cross-domain-component-summary} and the corresponding
benchmark-specific sections.


\section{Conclusion}
\label{sec:conclusion}

We introduced \textbf{TwinJEPA}, a JEPA-based framework for learning
\emph{action-preferred predictive representations} for zero-shot
goal-conditioned control.
TwinJEPA augments transition-level temporal prediction with offline-mined
local action-comparison supervision, constructing action-diverse pairs from
fixed offline data and using goal-conditioned reward relabeling to derive
reward-gap magnitude and preference-ordering targets.
These objectives encourage the learned representation to preserve
action-discriminative information not explicitly captured by transition-wise
latent prediction, while introducing no additional inference-time computation.
Through matched-seed evaluations across long-horizon navigation and
continuous-control benchmarks, we find that the effect of comparative
supervision is positive on average in several paper-aligned settings but
varies substantially across domains and tasks.
The strongest improvements appear in environments where the mined comparisons
provide informative distinctions between alternative actions, whereas other
settings exhibit larger uncertainty across seeds.
Further analyses show that reward-gap regression and preference
classification provide distinct and domain-dependent sources of supervision,
rather than universally interchangeable improvements.
These findings suggest that temporal prediction and action comparison capture
complementary aspects of control-relevant representation learning.
While JEPA-style objectives model how behaviors evolve over time,
offline preference supervision provides additional information about which
actions are more suitable under a given goal.

\subsection*{AI use statement}

Generative AI tools were used only for language editing and presentation,
including improving grammar, clarity, concision, and readability.
They were not used to develop the research methodology, formulate
mathematical claims, implement the proposed method, generate or process
experimental data, conduct experiments, or interpret the results.
All AI-assisted edits were reviewed and verified by the authors, who take
full responsibility for the final content of the paper.

\subsection*{Reproducibility statement}

We provide the methodological, implementation, and evaluation details needed
to reproduce the reported results in the main paper and appendix.
The main text specifies the TwinJEPA objective, local action-comparison
mining procedure, goal-conditioned relabeling strategy, auxiliary losses,
joint representation-learning objective, and deployment procedure.
It also defines the matched-seed evaluation protocol, including paired
differences, paired-bootstrap confidence intervals, and exact two-sided
sign-flip permutation tests.

Appendix~B.1 provides the complete state-based training procedure
(Algorithm~1), shared hyperparameters, hardware configurations, training
budgets, and evaluation environments.
Appendix~B.3 reports sensitivity to the local observation-matching tolerance
used during pair construction.
Benchmark-specific sections document the corresponding datasets, matching
parameters, training horizons, seed configurations, task-level evaluations,
and implementation choices.
Appendix~C.2 further verifies the reproduction of the Official TD-JEPA
baseline used for the matched AntMaze comparison.

Appendices~C--H provide the underlying task-level and seed-level results,
component ablations, auxiliary-weight and hyperparameter sensitivity analyses,
and benchmark-specific numerical breakdowns.
Appendix~H additionally specifies the pixel-based ExORL protocol, including
the observation construction, encoder setting, training configuration, and
matched-seed evaluation procedure.
Unless otherwise stated, each matched comparison uses the same backbone
architecture, offline data, optimization schedule, training budget, random
seeds, task-inference procedure, and evaluation protocol for Official TD-JEPA
and TwinJEPA, with the action-comparison objective as the systematic
difference between the two methods.
Finally, Appendix~I states the assumptions and provides the propositions and
proofs underlying the theoretical analysis.

\bibliography{iclr2027_conference}
\bibliographystyle{iclr2027_conference}
\clearpage
\appendix\label{appendix}
\setcounter{secnumdepth}{2}
\newcommand{\appsectionentry}[2]{%
  \makebox[3.2em][l]{\bfseries\ref{#1}}%
  \textbf{#2}
  & Page~\pageref{#1} \\[0.15em]
}

\newcommand{\appsubsectionentry}[2]{%
  \hspace*{1.5em}%
  \makebox[3.2em][l]{\ref{#1}}%
  #2
  & Page~\pageref{#1} \\
}

\begin{center}
    {\LARGE\bfseries Appendix}
\end{center}
\vspace{0.35em}
\hrule
\vspace{0.6em}

\noindent{\bfseries Appendix Overview}\par
\vspace{0.3em}

\begingroup
\normalsize
\setlength{\tabcolsep}{0pt}
\renewcommand{\arraystretch}{1.08}

\begin{tabular*}{\textwidth}
{@{\extracolsep{\fill}}lr@{}}

\appsectionentry
  {sec:related-work}
  {Extended Related Work} 

\appsectionentry
  {sec:additional-details}
  {Additional Experimental Details}

\appsubsectionentry
  {detail}
  {Implementation Details}

\appsubsectionentry
  {app:cross-domain-component-summary}
  {Cross-Domain Auxiliary-Objective Summary}

  \appsubsectionentry
{app:tol-sensitivity}
{Sensitivity to Local Matching Tolerance}

\appsectionentry
  {app:antmaze-extended}
  {AntMaze Extended Evaluation}

\appsubsectionentry
  {app:antmaze-task}
  {Task-Level and Domain-Level Summary}
  
\appsubsectionentry
{app:antmaze-reproduction}
{Official TD-JEPA Reproduction Check}

\appsubsectionentry
  {app:antmaze-ablation}
  {Additional Ablation Results}

\appsubsectionentry
  {app:mn-dynamics}
  {Learning Dynamics on Medium-Navigate}

\appsectionentry
  {app:dmc}
  {DMC Walker Extended Evaluation}

\appsubsectionentry
{app:walker-grid}
  {Hyperparameter-Grid Evaluation}

\appsectionentry
  {sec:cheetah}
  {Cheetah Offline Control on ExORL}

\appsubsectionentry
{app:cheetah-task-results}
{Complete Task-Level Results}

\appsubsectionentry
{app:cheetah-ablation}
{Auxiliary-Objective and Weight Sensitivity}

\appsectionentry
  {sec:quadruped}
  {Quadruped Extended Evaluation}

\appsubsectionentry  
{app:quad-seeds}
{Seed-Level Results}

\appsubsectionentry
{app:quad-task}
{Task-Level Numerical Breakdown}

\appsubsectionentry
{app:quadruped-ablation}
{Additional Auxiliary-Objective Sensitivity}

\appsectionentry
  {sec:maze2d-evaluation}
  {D4RL Maze2D Extended Evaluation}

\appsectionentry
  {app:rgb-obs}
  {Pixel-Based ExORL Evaluation}

\appsectionentry
  {sec:theory-twin}
  {Theoretical Properties of TwinJEPA}

\appsectionentry
  {sec:additional-results-summary}
  {Summary of Additional Results}

\end{tabular*}

\endgroup

\vspace{0.8em}

\section{Extended Related Work}
\label{sec:related-work}

\paragraph{Representation learning for reinforcement learning.}
A broad line of work studies how learned representations can improve
reinforcement learning by preserving task-relevant structure while discarding
nuisance variation.
Existing approaches include contrastive objectives, predictive
representation learning, prototype-based methods, and representation
objectives derived from behavioral or dynamical equivalence.
In offline and zero-shot RL, methods such as HILP~\cite{park2024foundation},
FB~\cite{touati2021learning}, and regularized latent-dynamics approaches
~\cite{jajoo2025regularized} learn reusable representations from fixed
behavioral data.
%
%
TwinJEPA is complementary to these approaches.
Rather than replacing the underlying representation objective, it augments a
temporal predictive representation with supervision obtained by comparing
different actions observed in locally similar states.

\paragraph{JEPA and predictive representations for control.}
Joint-Embedding Predictive Architectures learn representations by predicting
future latent variables instead of reconstructing observations
~\cite{lecun2022path,assran2023self,klindt2026does}.
Recent extensions have explored JEPA-style prediction for multimodal and
decision-making settings~\cite{chen2025vl,destrade2025value,vujinovic2026act}, while TD-JEPA
adapts latent temporal prediction to zero-shot offline control
~\cite{bagatella2025td}.
%
%
These methods primarily supervise representations through predictive
consistency along observed trajectories.
TwinJEPA addresses a different axis of representation learning by explicitly
introducing comparisons between alternative actions in local state
neighborhoods.
The proposed objective therefore complements temporal prediction rather than
replacing the JEPA backbone.

\paragraph{Preference and pairwise supervision.}
Preference-based reinforcement learning typically learns from pairwise
comparisons between trajectories, actions, or outcomes, often obtained from
human feedback or an external preference oracle~\cite{lee2026preference,zhang2024learning,hwang2023sequential}.
%
%
TwinJEPA differs in both the source and the role of the preference signal.
It does not require human-provided preferences or additional environment
interaction.
Instead, pairwise labels and reward-gap targets are generated automatically
from locally matched transitions through goal-conditioned reward relabeling.
Moreover, these comparisons are used as auxiliary supervision for
representation learning rather than as a standalone reward model or policy
optimization objective.
This distinction allows TwinJEPA to inject action-discriminative information
into a temporal predictive representation while retaining the original
offline control pipeline.

\section{Additional Experimental Details}
\label{sec:additional-details}

This appendix provides additional experimental analyses that complement the
main results. Specifically, we report implementation details, environment
descriptions, extended task-level evaluations, auxiliary-objective analyses,
hyperparameter sensitivity studies, and additional cross-domain evaluations.
Unless otherwise stated, official TD-JEPA and TwinJEPA use identical backbone
architectures, optimization hyperparameters, training budgets, random seeds,
and evaluation protocols within each matched comparison.

The primary difference between the two methods is that TwinJEPA augments the
TD-JEPA objective with offline-mined action-preference supervision.
Specifically, locally matched transition pairs are constructed from the fixed
offline dataset, and goal-conditioned reward relabeling provides two auxiliary
signals: reward-gap regression and preference classification.
These auxiliary objectives are applied only during representation learning and
introduce no additional inference-time computation.

\subsection{Implementation Details}
\label{detail}

We report task- and domain-level results on AntMaze, auxiliary-objective
ablations on AntMaze and ExORL Cheetah, Fixed-HP evaluations on DMC Walker,
ExORL Cheetah, ExORL/DMC Quadruped, and D4RL Maze2D, as well as additional
pixel-based evaluations on ExORL domains.

Full TwinJEPA with $\lambda=0.1$ serves as the primary matched configuration.
Unless stated otherwise, all comparisons use identical training and evaluation
protocols, and results are reported as mean $\pm$ standard deviation over
matched seeds. Task-level or seed-level breakdowns are provided when they
reveal information beyond aggregate performance.

The shared value of $\lambda$ is maintained across domains to evaluate the
transferability of the proposed supervision under a fixed protocol rather than
selecting benchmark-specific optimal weights after observing downstream
performance. Additional $\lambda$ sweeps are therefore treated as sensitivity
analyses and are not used to determine the primary reported configurations.

\begin{algorithm}[t]
\caption{TwinJEPA: Offline Pair Mining and Joint Training}
\label{alg:twinjepa}
\begin{algorithmic}[1]

\REQUIRE
Offline dataset $\mathcal{D}$,
matching function $q$,
resolution $\tau_{\mathrm{obs}}$,
action threshold $\epsilon_a$,
pair budget $N_{\mathrm{pair}}$,
training steps $T$

\ENSURE
Encoders $\phi,\psi$ and policy $\pi$

\STATE \textbf{Phase I: construct local action-comparison pairs}
\STATE Assign each transition
$b_i \leftarrow
\operatorname{round}
(q(\bar{o}_i)/\tau_{\mathrm{obs}})$
\STATE
Initialize
$\mathcal{D}_{\mathrm{pair}}\leftarrow\emptyset$

\FOR{candidate pairs $(i,j)$ satisfying $b_i=b_j$}
    \IF{$\|a_i-a_j\|_2\geq\epsilon_a$}
        \STATE Relabel the two successor states under the available goals
        \STATE
        $k^\star
        \leftarrow
        \arg\max_k
        |r_i^{(k)}-r_j^{(k)}|$
        \STATE
        $\Delta r
        \leftarrow
        |r_i^{(k^\star)}-r_j^{(k^\star)}|$
        \IF{$\Delta r>10^{-6}$}
            \STATE
            $y_{\mathrm{pref}}
            \leftarrow
            \mathbb{I}
            [r_j^{(k^\star)}>r_i^{(k^\star)}]$
            \STATE
            Add
            $(\bar{o}_i,a_i,a_j,\Delta r,y_{\mathrm{pref}})$
            to $\mathcal{D}_{\mathrm{pair}}$
        \ENDIF
    \ENDIF
    \IF{$|\mathcal{D}_{\mathrm{pair}}|\geq N_{\mathrm{pair}}$}
        \STATE \textbf{break}
    \ENDIF
\ENDFOR

\STATE \textbf{Phase II: joint representation learning}

\FOR{$t=1$ \TO $T$}
    \STATE
    Sample transition minibatch
    $\mathcal{B}\sim\mathcal{D}$
    and comparison minibatch
    $\mathcal{B}_{\mathrm{pair}}\sim\mathcal{D}_{\mathrm{pair}}$
    \STATE
    Compute
    $\mathcal{L}_{\mathrm{base}}$
    from $\mathcal{B}$
    \STATE
    Compute
    $\mathcal{L}_{\mathrm{pair}}$
    from $\mathcal{B}_{\mathrm{pair}}$
    \STATE
    Update representation parameters using
    $\mathcal{L}_{\mathrm{repr}}$
    \STATE
    Update $\pi$ using the unchanged FlowBC objective
    \STATE
    Update target networks by EMA
\ENDFOR

\RETURN $\phi,\psi,\pi$

\end{algorithmic}
\end{algorithm}

Experiments on DMC Walker and D4RL Maze2D were conducted using NVIDIA
GeForce RTX 4090 GPUs, whereas AntMaze, ExORL Cheetah, and ExORL/DMC
Quadruped experiments were conducted using NVIDIA GeForce RTX 5090 GPUs.
The GPU model affects computational throughput only; within each reported
comparison, official TD-JEPA and TwinJEPA were trained under identical
hardware settings, software environments, optimization schedules, and
evaluation protocols. No benchmark-specific model or loss-weight retuning was performed, except for
explicitly reported domain-level parameters such as the local observation
matching tolerance used during offline preference-pair construction.

\subsubsection{Evaluation environments}

We evaluate TwinJEPA across five environment families that cover complementary
offline control regimes, including long-horizon sparse-reward navigation,
dense continuous locomotion, high-dimensional posture regulation, and
simplified geometric goal reaching.

\begin{itemize}

\item \textbf{OGBench AntMaze.}

OGBench AntMaze, represented by tasks such as
\texttt{antmaze-medium-navigate}, is a long-horizon goal-reaching domain in
which a quadrupedal ant must traverse a continuous maze to reach designated
goals. It combines sparse task feedback, high-dimensional locomotion, and
long-range navigation, making it a demanding evaluation of whether an
offline-learned representation supports zero-shot goal-conditioned control.

We use AntMaze as the primary benchmark because it directly tests the central
motivation of TwinJEPA: whether local action-comparison supervision can
provide additional action-discriminative information beyond transition-wise
temporal prediction.

\item \textbf{DeepMind Control Suite Walker.}

Walker, represented by \texttt{walker\_walk}, is a proprioceptive
continuous-control domain requiring a bipedal agent to maintain balance and
generate coordinated locomotion. Unlike AntMaze, Walker does not involve
explicit maze geometry or sparse goal reaching; instead, performance depends
on stable posture control and dynamic coordination.

Walker therefore evaluates whether the proposed reward-gap and preference
objectives remain useful in dense-reward continuous-control settings beyond
sparse navigation.

\item \textbf{DeepMind Control Suite Cheetah.}

Cheetah, represented by \texttt{cheetah\_run}, is a planar locomotion domain
with dense and continuously varying rewards. It differs from Walker in
embodiment and control dynamics while preserving the need for coordinated
movement.

We include Cheetah to examine whether reward-gap regression and preference
classification provide transferable supervisory signals across different
continuous-control systems rather than benefiting only a particular
locomotion morphology.

\item \textbf{DeepMind Control Suite Quadruped.}

Quadruped, represented by \texttt{quadruped\_stand}, is a high-dimensional
proprioceptive control domain in which a four-legged agent must establish and
maintain stable posture. Compared with Walker and Cheetah, the task emphasizes
joint coordination, body orientation, and local stability rather than sustained
displacement.

Quadruped provides a complementary evaluation setting where the effectiveness
of local action comparisons can be examined under complex embodiment but
relatively localized control objectives.

\item \textbf{D4RL Maze2D.}

Maze2D, represented by \texttt{maze2d-medium}, is a planar point-mass
navigation domain in which an agent moves through a geometric maze toward a
target location. Compared with AntMaze, it removes the complexity of
quadrupedal locomotion and isolates the geometric navigation component of
goal-conditioned control.

Maze2D therefore serves as a lower-complexity counterpart to AntMaze and
provides an additional test of whether TwinJEPA remains compatible with the
TD-JEPA framework when embodiment and dynamics are simplified.

\end{itemize}

These environment families provide complementary evaluation conditions.
AntMaze and Maze2D examine goal-directed navigation under different levels of
embodiment complexity. Walker and Cheetah evaluate dynamic locomotion with
dense reward feedback, whereas Quadruped studies high-dimensional posture
control. Together, these evaluations are intended not only to measure whether TwinJEPA
improves average performance, but also to characterize when locally mined
action comparisons provide useful representation-level supervision. The
contribution of preference supervision may depend on factors including reward
structure, task horizon, action diversity, and the quality of the mined local
comparison distribution.

\subsection{Cross-Domain Auxiliary-Objective Summary}
\label{app:cross-domain-component-summary}

Section~\ref{sec:component-analysis} summarizes how reward-gap regression
and preference classification contribute across environments.
For completeness, Table~\ref{tab:cross-domain-component-summary} collects
the component ablations used in that comparison under their corresponding
matched protocols.
The table is intended as a cross-domain summary; complete within-domain
sensitivity analyses, uncertainty estimates, and seed-level results are
reported in the benchmark-specific appendix sections.

\begin{table}[t]
\centering
\caption{
Cross-domain auxiliary-objective summary.
AntMaze entries report zero-shot success (\%, mean$\pm$SE over 10 matched
seeds), whereas Cheetah and Quadruped report four-task mean returns
(mean$\pm$standard deviation).
Full denotes
$(\lambda,\alpha_{\mathrm{gap}},\alpha_{\mathrm{pref}})
=(0.1,1.0,0.5)$;
Gap-only removes preference classification and Pref-only removes reward-gap
regression.
$^{\dagger}$Cheetah component ablations use the matched five-seed ablation
subset and should not be compared directly with the ten-seed primary
Cheetah result.
}
\label{tab:cross-domain-component-summary}
{\small
\setlength{\tabcolsep}{4.2pt}
\renewcommand{\arraystretch}{1.05}

\begin{tabular}{@{}lcccc@{}}
\toprule
Setting
& TD-JEPA
& Full
& Gap-only
& Pref-only \\
\midrule

AntMaze mn
& $68.8\!\pm\!4.0$
& $73.2\!\pm\!3.5$
& $74.0\!\pm\!3.0$
& $77.0\!\pm\!2.9$ \\

AntMaze ms
& $51.4\!\pm\!3.7$
& $52.0\!\pm\!4.3$
& $57.4\!\pm\!2.4$
& $50.2\!\pm\!3.4$ \\

Cheetah$^{\dagger}$
& $610.5\!\pm\!44.0$
& $671.4\!\pm\!32.4$
& $654.5\!\pm\!34.5$
& $629.8\!\pm\!37.2$ \\

Quadruped
& $559.4\!\pm\!35.4$
& $599.8\!\pm\!28.2$
& $554.4\!\pm\!18.2$
& $548.4\!\pm\!26.8$ \\

\bottomrule
\end{tabular}
}

\end{table}

The component ordering varies across environments.
Preference-only has the highest mean on AntMaze medium-navigate, whereas
gap-only has the highest mean on medium-stitch.
On Cheetah, gap-only retains much of the increase obtained by the full
objective, while preference-only produces a smaller positive difference.
Quadruped exhibits a different pattern, with the full objective exceeding
both isolated components under the evaluated configuration.
These results therefore support a domain-dependent interpretation of the
two auxiliary signals rather than a universal ordering or a claim that
their combination must dominate in every setting.

The corresponding detailed analyses are reported in
Section~\ref{app:antmaze-ablation} for AntMaze,
Section~\ref{app:cheetah-ablation} for Cheetah, and
Section~\ref{app:quadruped-ablation} for Quadruped.

\subsection{Sensitivity to Local Matching Tolerance}
\label{app:tol-sensitivity}

The neighborhood resolution $\tau_{\mathrm{obs}}$ controls the granularity
of the local state matches used to construct comparison pairs.
To examine whether the reported results depend on a narrowly tuned
resolution, we vary $\tau_{\mathrm{obs}}$ around the default operating point
while keeping the remaining training and evaluation settings unchanged.
These sweeps are used only as sensitivity diagnostics and do not determine
the default configurations used in the main experiments.

\begin{table}[t]
\centering
\caption{
Sensitivity to the local observation-matching tolerance
$\tau_{\mathrm{obs}}$.
Entries report zero-shot success (\%, mean$\pm$standard deviation) for the
currently available tolerance-probe runs.
Relative scale is defined with respect to the default tolerance used for
each domain.
$^\dagger$ denotes an incomplete seed set.
The probe results are diagnostic and are not used for model selection or
for the primary performance comparisons.
}
\label{tab:tol-sensitivity}
\small
\setlength{\tabcolsep}{5.0pt}
\renewcommand{\arraystretch}{1.05}

\begin{tabular}{lcccc}
\toprule
Domain & $\tau_{\mathrm{obs}}$ & Relative scale
& Success (\%) & $n$ \\
\midrule

\multicolumn{5}{l}{\textit{AntMaze}} \\
& 0.15 & $0.5\times$ & $76.7 \pm 6.1$ & $3^\dagger$ \\
& 0.30 & $\mathbf{1.0\times}$ & $76.0 \pm 2.5$ & 5 \\
& 0.60 & $2.0\times$ & $78.0 \pm 3.5$ & $3^\dagger$ \\

\midrule

\multicolumn{5}{l}{\textit{Maze2D exploratory probe}} \\
& 0.05 & $0.33\times$ & $8.0 \pm 0.0$ & $2^\dagger$ \\
& 0.10 & $0.67\times$ & $13.0 \pm 4.2$ & $2^\dagger$ \\
& 0.15 & $\mathbf{1.0\times}$ & $8.0 \pm 0.0$ & $2^\dagger$ \\
& 0.30 & $2.0\times$ & $6.0 \pm 2.8$ & $2^\dagger$ \\
& 0.50 & $3.33\times$ & $13.0 \pm 7.1$ & $2^\dagger$ \\

\bottomrule
\end{tabular}
\end{table}

On AntMaze, success remains within a narrow range of
$76.0$--$78.0\%$ when the matching tolerance is varied from
$0.5\times$ to $2\times$ the default value.
In particular, the default $\tau_{\mathrm{obs}}=0.30$ is not an isolated
high-performing operating point, providing no indication that the primary
configuration relies on a sharply tuned neighborhood resolution.

The currently available Maze2D runs show a less stable and non-monotonic
pattern across a substantially wider tolerance range.
However, only two seeds are presently available for each setting, so these
values are treated as an exploratory diagnostic rather than evidence for a
cross-domain robustness claim.
Importantly, the sweep does not reveal a systematic advantage uniquely at
the default tolerance.

Overall, the available sensitivity results are consistent with
$\tau_{\mathrm{obs}}$ acting as a domain-scale-dependent pair-construction
parameter rather than a narrowly optimized downstream-performance
hyperparameter.
We therefore retain the prespecified default tolerances used in the primary
experiments rather than selecting a value post hoc from these sweeps.

\section{AntMaze Extended Evaluation}
\label{app:antmaze-extended}

Unless stated otherwise, all results in this section use the shared Fixed-HP
configuration
$(\psi_{\mathrm{orth}},\mathrm{lr}_{\psi})=(0.1,10^{-4})$,
$10$ matched seeds, and evaluation after $1\mathrm{M}$ training steps.
We provide additional task-level statistics and reproduction details for
the OGBench AntMaze evaluation reported in the main paper.
AntMaze is used as the primary benchmark because it combines long-horizon
goal-reaching, sparse feedback, and high-dimensional locomotion.

\begin{figure*}[ht]
\centering
\includegraphics[width=0.85\linewidth]
{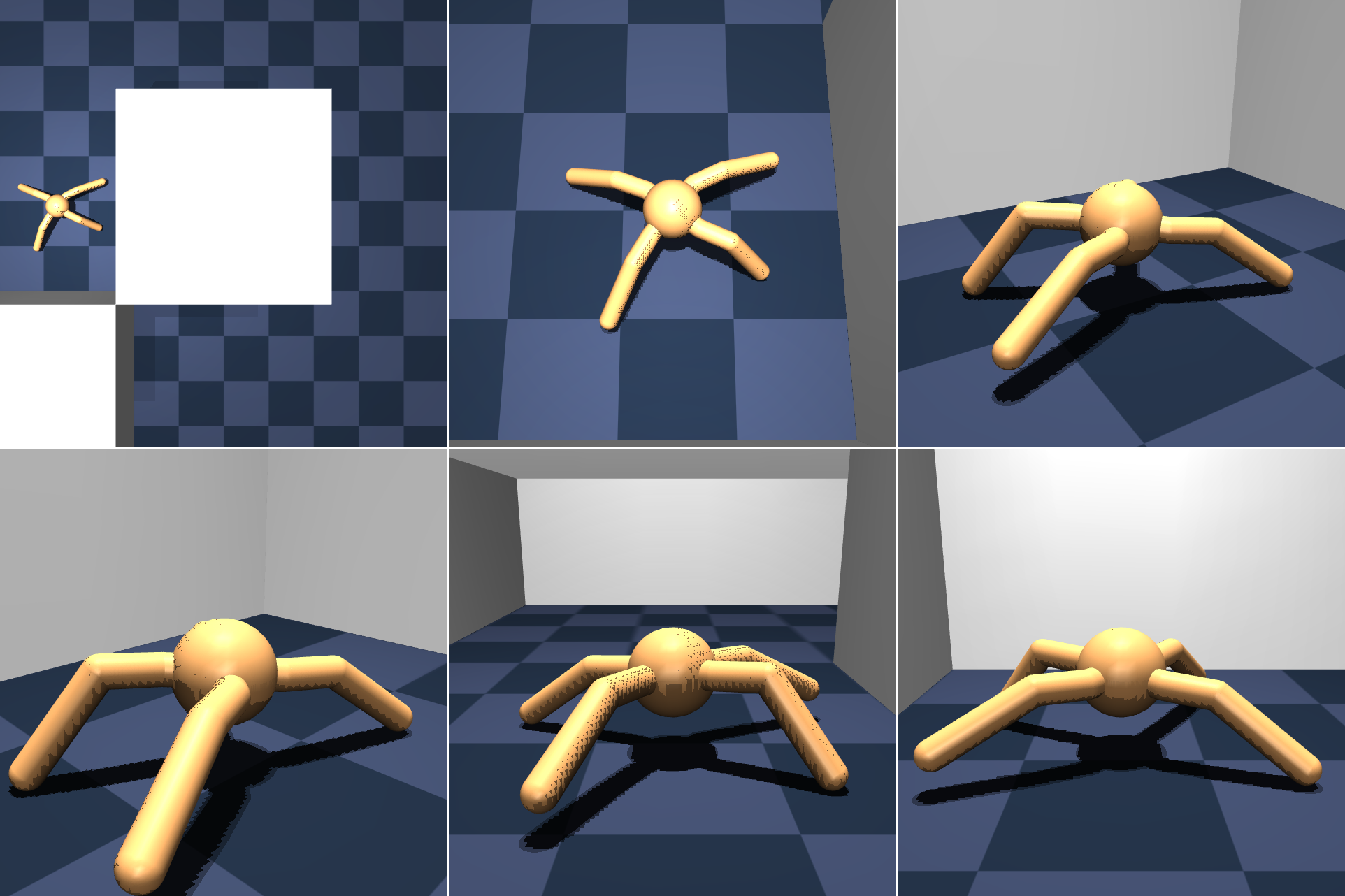}
\caption{
Environment visualization of the OGBench
\texttt{antmaze-medium-navigate-v0} environment from six simulator
viewpoints.
The views illustrate the long-horizon navigation structure evaluated in the
AntMaze experiments.
}
\label{fig:antmaze-views}
\end{figure*}

\subsection{Task-Level AntMaze Statistics}
\label{app:antmaze-task}

The main paper reports domain-level AntMaze results and their matched-seed
uncertainty analysis.
Here we provide the complete task-level decomposition over the four evaluated
domains and five evaluation goals per domain.
Each entry reports the zero-shot success rate averaged over the same matched
seeds used in the main evaluation.

\begin{table}[ht]
\centering
\caption{
Complete task-level zero-shot success rates (\%) on the four AntMaze domains.
The $\Delta$ row denotes the difference
$\mathrm{TwinJEPA}-\mathrm{Official}$.
The Avg column corresponds to the unweighted average over five goals within
each domain, while AntMaze-4 reports the average across the four domain
means.
}
\label{tab:antmaze-pertask-compact}
\setlength{\tabcolsep}{6pt}
\renewcommand{\arraystretch}{1.05}

\begin{tabular}{@{}llcccccc@{}}
\toprule
Domain & Method & T1 & T2 & T3 & T4 & T5 & Avg \\
\midrule

mn & Official
& $81$ & $80$ & $32$ & $57$ & $94$ & $68.8$ \\

& TwinJEPA
& $79$ & $86$ & $49$ & $54$ & $98$ & $73.2$ \\

& $\Delta$
& $-2$ & $+6$ & $+17$ & $-3$ & $+4$ & $+4.4$

\\
\addlinespace

ln & Official
& $76$ & $39$ & $51$ & $16$ & $29$ & $42.2$ \\

& TwinJEPA
& $81$ & $25$ & $55$ & $9$ & $45$ & $43.0$ \\

& $\Delta$
& $+5$ & $-14$ & $+4$ & $-7$ & $+16$ & $+0.8$

\\
\addlinespace

ms & Official
& $52$ & $3$ & $59$ & $81$ & $62$ & $51.4$ \\

& TwinJEPA
& $65$ & $1$ & $41$ & $78$ & $75$ & $52.0$ \\

& $\Delta$
& $+13$ & $-2$ & $-18$ & $-3$ & $+13$ & $+0.6$

\\
\addlinespace

ls & Official
& $44$ & $18$ & $76$ & $47$ & $4$ & $37.8$ \\

& TwinJEPA
& $55$ & $15$ & $80$ & $45$ & $11$ & $41.2$ \\

& $\Delta$
& $+11$ & $-3$ & $+4$ & $-2$ & $+7$ & $+3.4$

\\

\midrule

\multicolumn{2}{@{}l}{AntMaze-4}
& \multicolumn{5}{c}{Official $50.1$~~~TwinJEPA $52.4$}
& $+2.3$

\\

\bottomrule
\end{tabular}

\end{table}

\begin{table}[t]
\centering
\caption{
Complete task-level zero-shot success rates (\%) on the four AntMaze
domains.
The $\Delta$ row denotes
$\mathrm{TwinJEPA}-\mathrm{Official}$.
The Avg column is the unweighted average over the five evaluation goals,
while AntMaze-4 reports the average across the four domain means.
}
\label{tab:antmaze-task-level}
\end{table}

Figure~\ref{fig:antmaze-diagnostics} complements the complete numerical
decomposition with absolute and difference-based views of the same matched
evaluation.
Panel~(a) shows that the positive AntMaze-4 mean difference is composed of
modest domain-level changes rather than a large shift on a single domain.
Panel~(b) resolves medium-navigate into its five evaluation goals, where the
largest positive difference occurs on Task~3.
Panel~(c) extends this view to all 20 AntMaze goals and makes clear that
positive and negative task-level changes coexist within each domain.

\begin{figure}[t]
\centering
\includegraphics[width=\linewidth]
{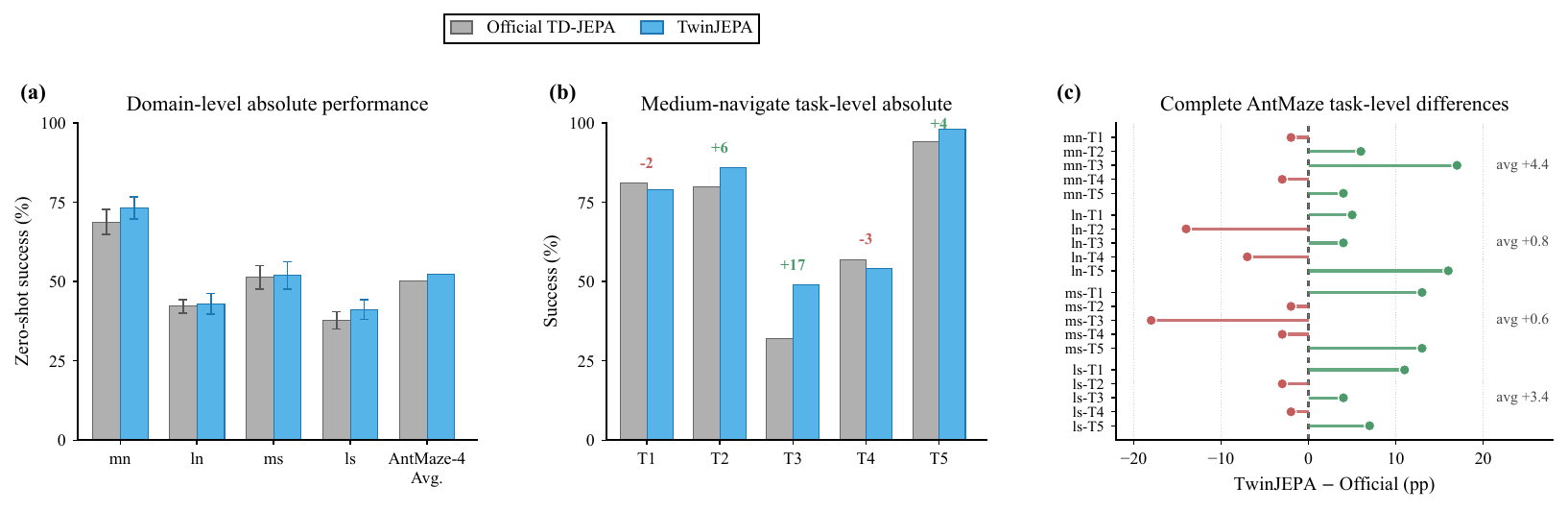}
\caption{
Absolute and task-level performance decomposition on OGBench AntMaze under
the shared Fixed-HP protocol.
\textbf{(a)} Domain-level zero-shot success of Official TD-JEPA and
TwinJEPA; error bars denote SE over 10 matched seeds where applicable.
\textbf{(b)} Absolute per-task success on medium-navigate, with annotations
showing
$\Delta=\mathrm{TwinJEPA}-\mathrm{Official}$
in percentage points.
\textbf{(c)} Complete task-level differences across the four AntMaze
domains; annotations on the right report the corresponding domain means.
Together, the panels expose the absolute performance levels underlying the
difference-based analysis in the main paper and show that the positive
domain-level means arise from heterogeneous goal-level changes.
}
\label{fig:antmaze-diagnostics}
\end{figure}

The decomposition also clarifies why the aggregate AntMaze result should not
be interpreted as a uniform improvement.
For example, medium-navigate combines a $+17$ point difference on Task~3
with small negative differences on Tasks~1 and~4, while large-navigate and
medium-stitch contain both sizable positive and negative changes.
The resulting domain means are nevertheless positive in all four cases.
This pattern is consistent with the main-paper conclusion that local
action-comparison supervision has task-dependent effects rather than shifting
all goals in the same direction.
\subsection{Official TD-JEPA Reproduction Check}
\label{app:antmaze-reproduction}

Before evaluating TwinJEPA, we verify that our matched implementation of
Official TD-JEPA reproduces the expected AntMaze performance regime.
Under the Fixed-HP protocol, our medium-navigate reproduction achieves
$68.8\pm4.0\%$ success, close to the reported TD-JEPA result of $70.4\%$.
This agreement indicates that the reproduced baseline follows the intended
training and evaluation setup despite the expected variation across
independent runs.

For the primary comparison, we use this reproduced implementation rather than
the published point estimate as the matched baseline. Official TD-JEPA and
TwinJEPA are trained with the same offline data, backbone architecture,
optimization schedule, training budget, task-inference procedure, evaluation
protocol, and matched random seeds. The only systematic difference is the
additional offline-mined action-comparison supervision introduced by
TwinJEPA. This controlled setup avoids attributing cross-paper differences in
training or evaluation to the proposed objective.

\begin{table}[t]
\centering
\caption{
Domain-level zero-shot success (\%) on AntMaze under the Fixed-HP protocol.
Results are reported as mean$\pm$SE over 10 matched seeds.
}
\label{tab:paper-baseline}
\setlength{\tabcolsep}{3.5pt}

\begin{tabular}{lcc}
\toprule
Setting & Official TD-JEPA & TwinJEPA \\
\midrule

mn
& $68.8{\pm}4.0$
& $\mathbf{73.2{\pm}3.5}$ \\

ln
& $42.2{\pm}2.1$
& $\mathbf{43.0{\pm}3.3}$ \\

ms
& $51.4{\pm}3.7$
& $\mathbf{52.0{\pm}4.3}$ \\

ls
& $37.8{\pm}2.7$
& $\mathbf{41.2{\pm}3.1}$ \\

\midrule

AntMaze-4 Avg.
& $50.1$
& $\mathbf{52.4}$ ($+2.3$)

\\

\bottomrule
\end{tabular}

\end{table}

Table~\ref{tab:paper-baseline} reports the resulting matched AntMaze
comparison across all four domains. The reproduced TD-JEPA results serve as
the reference for all within-paper effect estimates and uncertainty analyses.
Accordingly, differences between the published and reproduced TD-JEPA values
are treated as protocol- and run-level variation rather than as algorithmic
changes; the empirical effect of TwinJEPA is assessed only against the
matched reproduction.

\subsection{Additional Ablation and Sensitivity Results}
\label{app:antmaze-ablation}

The main paper summarizes the cross-domain contributions of reward-gap
regression and preference classification.
Here we provide the corresponding AntMaze sensitivity results and complete
numerical values under the same Fixed-HP protocol.
All experiments in this subsection use $10$ matched seeds and evaluation
after $1\mathrm{M}$ training steps.

\paragraph{Ablation protocol.}
Full TwinJEPA uses
$
(\lambda,\alpha_{\mathrm{gap}},\alpha_{\mathrm{pref}})
=
(0.1,1.0,0.5).
$
We consider two types of ablations.
First, we vary the auxiliary weight $\lambda$ while retaining both
comparative objectives.
Second, we isolate the reward-gap and preference terms by setting one of
their coefficients to zero.
These experiments are intended to characterize sensitivity and component
behavior rather than to select a benchmark-specific configuration for the
primary evaluation.

\begin{figure}[t]
\centering
\includegraphics[width=\linewidth]{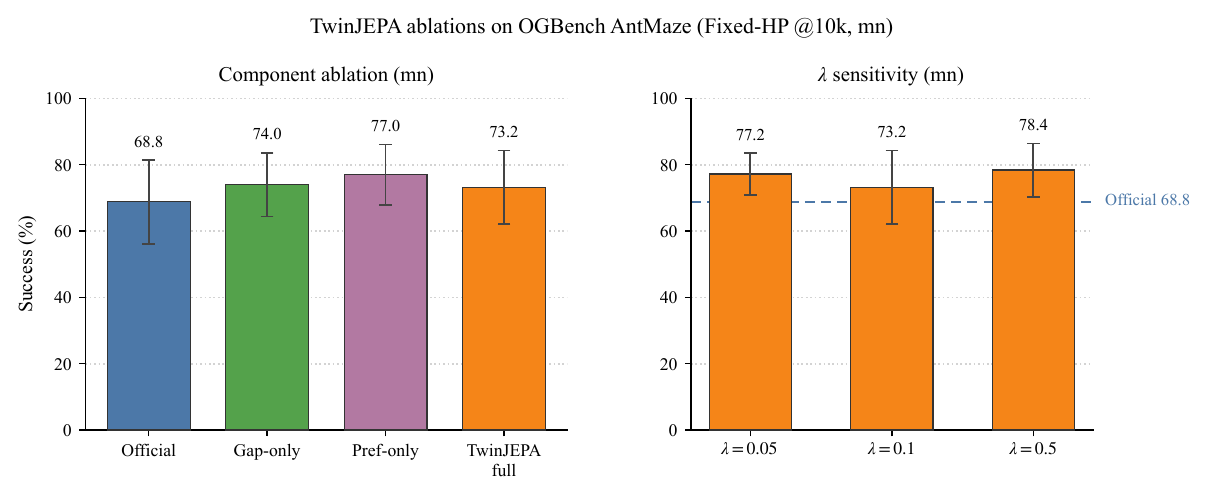}
\caption{
Additional TwinJEPA ablation and sensitivity results on AntMaze
medium-navigate under the Fixed-HP protocol using 10k inference samples.
The left panel isolates reward-gap regression and preference classification,
where Gap-only removes preference classification and Pref-only removes
reward-gap regression.
The right panel varies the auxiliary loss weight $\lambda$ while retaining
both objectives.
The shared setting $\lambda=0.1$ is used for the primary cross-domain
evaluation rather than selected as the optimum for medium-navigate.
Results are averaged over 10 matched seeds.
}
\label{fig:ablation_mn_fixedhp10k}
\end{figure}
%
%
%
%

\begin{table}[t]
\centering
\caption{
Additional AntMaze ablation and sensitivity results under the Fixed-HP
protocol (mean$\pm$SE over 10 matched seeds).
The upper block varies the auxiliary weight $\lambda$ on medium-navigate
while keeping both comparative objectives active.
The lower block isolates reward-gap regression and preference classification
on medium-navigate and medium-stitch.
Gap-only uses
$(\alpha_{\mathrm{gap}},\alpha_{\mathrm{pref}})=(1,0)$,
whereas Pref-only uses $(0,0.5)$.
The shared $\lambda=0.1$ configuration is used in the primary evaluation and
is not selected as the optimum for individual AntMaze domains.
}
\label{tab:ablation-compact}
\setlength{\tabcolsep}{3.2pt}
\renewcommand{\arraystretch}{1.05}

\begin{tabular}{llcc}
\toprule
Domain & Variant & Success (\%) & $\Delta$ vs.\ Off. \\
\midrule

\multicolumn{4}{l}{\textit{$\lambda$ sensitivity on medium-navigate}} \\

mn & Official
& $68.8{\pm}4.0$
& -- \\

mn & $\lambda=0.05$
& $77.2{\pm}2.0$
& $+8.4$ \\

mn & $\lambda=0.1$ (shared)
& $73.2{\pm}3.5$
& $+4.4$ \\

mn & $\lambda=0.5$
& $\mathbf{78.4{\pm}2.5}$
& $\mathbf{+9.6}$ \\

\midrule

\multicolumn{4}{l}{\textit{Component ablations at $\lambda=0.1$}} \\

mn & Gap-only
& $74.0{\pm}3.0$
& $+5.2$ \\

mn & Pref-only
& $\mathbf{77.0{\pm}2.9}$
& $\mathbf{+8.2}$ \\

mn & Full TwinJEPA
& $73.2{\pm}3.5$
& $+4.4$ \\

\addlinespace

ms & Official
& $51.4{\pm}3.7$
& -- \\

ms & Gap-only
& $\mathbf{57.4{\pm}2.4}$
& $\mathbf{+6.0}$ \\

ms & Pref-only
& $50.2{\pm}3.4$
& $-1.2$ \\

ms & Full TwinJEPA
& $52.0{\pm}4.3$
& $+0.6$ \\

\bottomrule
\end{tabular}

\end{table}

\paragraph{Sensitivity to the auxiliary weight.}
The medium-navigate sweep shows that the effect is not confined to a single
choice of $\lambda$.
All three evaluated weights produce higher mean success than the matched
Official TD-JEPA baseline, although their magnitudes differ.
In particular, the shared setting $\lambda=0.1$ is not the
best-performing value on medium-navigate in isolation.
We retain it in the primary experiments because the main evaluation uses a
common auxiliary weighting across domains rather than selecting
benchmark-specific values after observing downstream performance.
The sweep should therefore be interpreted as a sensitivity analysis rather
than as model selection for the reported main result.

\paragraph{Component behavior.}
The component ordering differs between the two AntMaze domains.
Preference-only gives the highest mean success on medium-navigate, whereas
gap-only gives the highest mean success on medium-stitch.
The full objective does not dominate both isolated components in either
case.
These results are consistent with the cross-domain analysis in the main
paper and show that reward-gap regression and preference classification
provide distinct, domain-dependent forms of supervision rather than a
universally synergistic combination.

The corresponding large-navigate and large-stitch variants remain within the
same general performance range as the matched TD-JEPA and full TwinJEPA
configurations and do not change the interpretation above.

\subsection{Learning Dynamics on Medium-Navigate}
\label{app:mn-dynamics}

\begin{figure*}[ht]
\centering
\includegraphics[width=\linewidth]{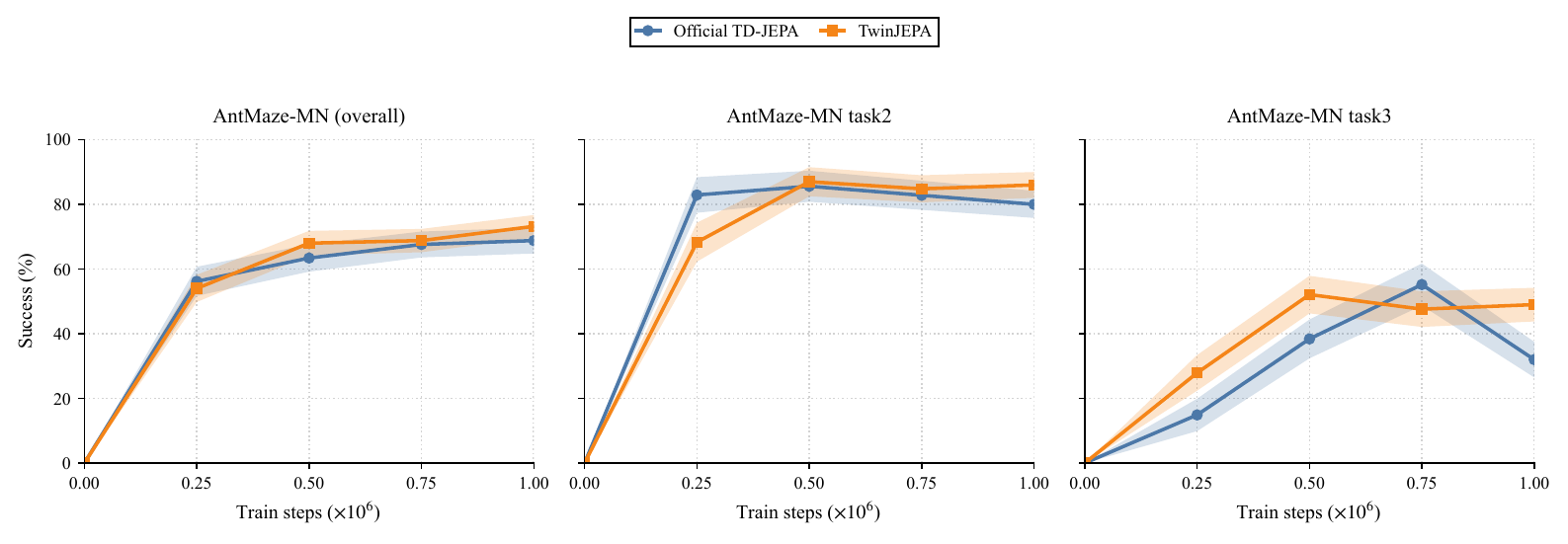}
\caption{
Learning dynamics on \texttt{antmaze-medium-navigate-v0} under the Fixed-HP
configuration
$(\psi_{\mathrm{ortho}},\mathrm{lr}_{\psi})=(0.1,10^{-4})$.
Curves report mean zero-shot success (\%) $\pm$SE over 10 matched seeds at
$\{0,0.25,0.5,0.75,1.0\}\times10^{6}$ training steps.
The left panel shows success averaged over the five evaluation goals
(final $68.8\%\to73.2\%$, $\Delta{+}{4.4}$).
The middle panel shows Task~2, a relatively high-performing goal
(final $80\%\to86\%$, $\Delta{+}{6}$).
The right panel shows Task~3, a lower-performing goal with the largest final
task-level difference
(final $32\%\to49\%$, $\Delta{+}{17}$).
}
\label{fig:mn-learning-curves}
\end{figure*}

\paragraph{Aggregate learning dynamics.}
Figure~\ref{fig:mn-learning-curves} provides a temporal view of the
medium-navigate result reported in the main evaluation.
At $250\mathrm{k}$ steps, Official TD-JEPA has a small advantage in the
five-task mean ($56.2\%$ versus $54.0\%$).
TwinJEPA is higher at $500\mathrm{k}$ steps
($68.0\%$ versus $63.4\%$), after which the difference narrows at
$750\mathrm{k}$ steps
($68.8\%$ versus $67.6\%$).
At the final checkpoint, TwinJEPA reaches $73.2\%$ compared with
$68.8\%$ for Official TD-JEPA.
The separation is therefore not monotonic over training, and the final
difference should not be interpreted as a persistent advantage at every
checkpoint.

\paragraph{Task-specific dynamics.}
The representative task curves illustrate two different temporal patterns.
On Task~2, both methods reach relatively high success early in training and
finish at $80\%$ for Official TD-JEPA and $86\%$ for TwinJEPA.
The difference is therefore modest relative to their overall performance.

Task~3 exhibits greater variation.
TwinJEPA first attains a mean success of at least $50\%$ at the
$500\mathrm{k}$ checkpoint, whereas Official TD-JEPA first reaches this
threshold at $750\mathrm{k}$.
Official TD-JEPA then decreases from approximately $57\%$ at
$750\mathrm{k}$ to $32\%$ at the final checkpoint, while TwinJEPA finishes
at $49\%$.
Thus, the final $17$-point difference reflects both earlier separation during
training and a larger late-stage decline of the baseline.
This decomposition is important because the final task-level gap should not
be attributed solely to a uniformly increasing TwinJEPA advantage. Since only five checkpoints are recorded, these curves do not establish
continuous-time convergence, monotonic improvement, or a formal difference
in optimization speed.
They instead provide a coarse temporal diagnostic showing how the final
medium-navigate differences emerge over the recorded training trajectory. The learning-curve analysis is complementary to the final matched-seed
statistics reported in the main paper.
It provides information about the temporal evolution of the observed
differences but is not used as an additional statistical test of the final
performance effect.

\begin{table}[t]
\centering
\small
\caption{
Checkpoint-level summary for
\texttt{antmaze-medium-navigate-v0} under the Fixed-HP protocol
(10 matched seeds).
The second column reports the first recorded checkpoint at which the mean
zero-shot success reaches at least $50\%$.
Because evaluations are available only at
$\{0,0.25,0.5,0.75,1.0\}\times10^{6}$ training steps, these values describe
checkpoint attainment rather than continuous convergence or sample
efficiency.
The remaining columns report final zero-shot success (\%) at
$1\mathrm{M}$ training steps.
}
\label{tab:mn-learning-efficiency}
\setlength{\tabcolsep}{5pt}
\renewcommand{\arraystretch}{1.05}

\begin{tabular}{lcccc}
\toprule
Setting
& First checkpoint $\geq 50\%$ (Off. / Twin)
& Official @ $1\mathrm{M}$
& TwinJEPA @ $1\mathrm{M}$
& $\Delta$
\\
\midrule

Overall (5-task avg.)
& $250\mathrm{k}$ / $250\mathrm{k}$
& $68.8$
& $\mathbf{73.2}$
& $+4.4$
\\

Task~3
& $750\mathrm{k}$ / $500\mathrm{k}$
& $32$
& $\mathbf{49}$
& $+17.0$
\\

Task~2
& $250\mathrm{k}$ / $250\mathrm{k}$
& $80$
& $\mathbf{86}$
& $+6.0$
\\

\bottomrule
\end{tabular}

\end{table}

\section{DMC Walker Extended Evaluation}
\label{app:dmc}

This section provides the extended Walker results that complement the
Fixed-HP comparison reported in the main paper.
While the main paper focuses on the matched-seed effect and task-level
decomposition, here we report the broader optimization grid used to assess
whether the observed Walker behavior is specific to a single hyperparameter
configuration.

\begin{figure*}[ht]
\centering
\includegraphics[width=0.85\linewidth]
{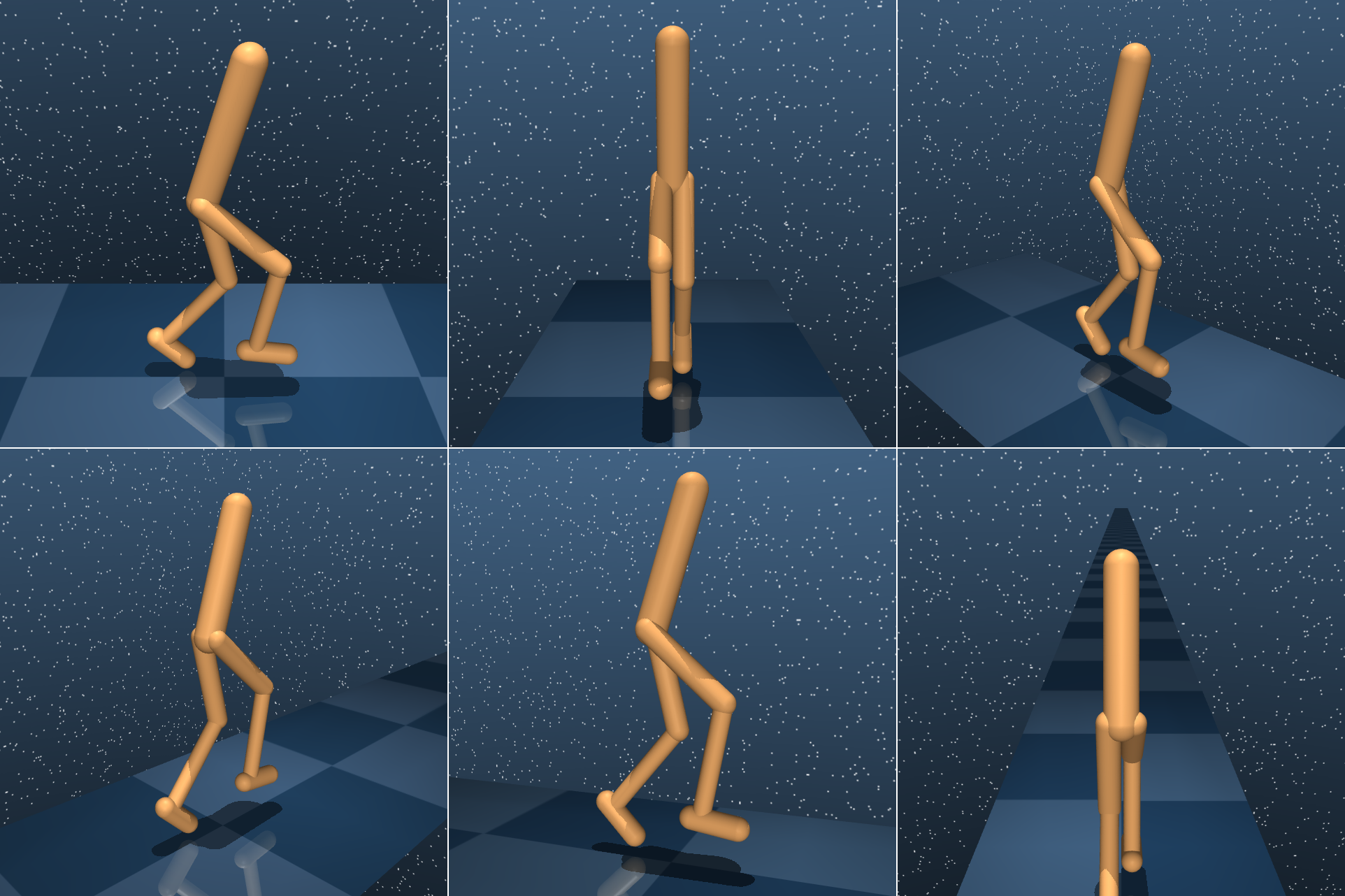}
\caption{
Environment visualizations of the DeepMind Control Suite
\texttt{walker-walk} environment from six camera viewpoints
(2$\times$3).
Top row: side, front, and three-quarter front views.
Bottom row: three-quarter rear, low-angle, and back views.
}
\label{fig:walker-walk-views}
\end{figure*}

\subsection{Hyperparameter-Grid Evaluation}
\label{app:walker-grid}

We evaluate the ExORL/DMC Walker proprioceptive benchmark over the same
six-setting grid used for the TD-JEPA comparison,
$
\psi_{\mathrm{ortho}}\in\{0.01,0.1,1\},
\qquad
\mathrm{lr}_{\psi}\in\{10^{-4},10^{-5}\}.
$
For each setting, performance is averaged over
\texttt{stand}, \texttt{walk}, \texttt{run}, and \texttt{spin}
after $3\,\mathrm{M}$ training steps using five matched seeds.
The paper-aligned Fixed-HP configuration corresponds to
$(\psi_{\mathrm{ortho}},\mathrm{lr}_{\psi})=(0.1,10^{-4})$. Three TwinJEPA runs that exhibited training instability were rerun using the
same seed and hyperparameter configuration and the rerun values are used in
the reported comparison
(seed~$8948$ at $(0.01,10^{-5})$ and $(0.1,10^{-4})$,
and seed~$4729$ at $(1,10^{-5})$).
The matched seed set itself is unchanged.

\begin{table}[t]
\centering
\caption{
Hyperparameter-grid evaluation on ExORL/DMC Walker.
Each row reports the mean episode return averaged over
\texttt{stand}, \texttt{walk}, \texttt{run}, and \texttt{spin}
after $3\,\mathrm{M}$ training steps.
Individual configurations are reported as mean$\pm$SE over five matched seeds.
The final row aggregates the six settings, corresponding to 30 matched
seed--configuration pairs.
$\Delta$ denotes
$\mathrm{TwinJEPA}-\mathrm{Official\ TD\text{-}JEPA}$.
The shared main-paper configuration
$(\psi_{\mathrm{ortho}},\mathrm{lr}_{\psi})=(0.1,10^{-4})$
is shown as one point in this broader grid rather than selected as the
best-performing Walker configuration.
}
\label{tab:dmc-full}
\setlength{\tabcolsep}{3.5pt}
\renewcommand{\arraystretch}{1.05}

\begin{tabular}{ccccc}
\toprule
$\psi_{\mathrm{ortho}}$
& $\mathrm{lr}_{\psi}$
& Official TD-JEPA
& TwinJEPA
& $\Delta$
\\
\midrule

$0.01$ & $10^{-4}$
& $779.1{\pm}8.6$
& $779.6{\pm}8.4$
& $+0.5$
\\

$0.01$ & $10^{-5}$
& $766.2{\pm}16.7$
& $763.7{\pm}9.0$
& $-2.5$
\\

$0.1$ & $10^{-4}$
& $774.3{\pm}10.6$
& $782.7{\pm}4.7$
& $+8.4$
\\

$0.1$ & $10^{-5}$
& $754.9{\pm}9.8$
& $767.5{\pm}7.2$
& $+12.6$
\\

$1$ & $10^{-4}$
& $771.9{\pm}8.9$
& $776.4{\pm}6.0$
& $+4.5$
\\

$1$ & $10^{-5}$
& $746.8{\pm}15.0$
& $757.6{\pm}15.3$
& $+10.8$
\\

\midrule

\multicolumn{2}{l}{Overall (30 pairs)}
& $765.5{\pm}4.9$
& $771.3{\pm}3.8$
& $+5.8$
\\

\bottomrule
\end{tabular}

\end{table}

\paragraph{Grid-level behavior.}
Table~\ref{tab:dmc-full} shows that the Walker effect varies with the
underlying TD-JEPA optimization configuration.
TwinJEPA has a higher mean return in five of the six settings and a lower
mean in one setting, with differences ranging from $-2.5$ to $+12.6$.
At the shared Fixed-HP point used in the main evaluation, the four-task mean
changes from $774.3$ to $782.7$ ($\Delta{+}{8.4}$).
Across all six settings, the aggregate mean changes from
$765.5$ to $771.3$ ($\Delta{+}{5.8}$).

The grid is intended as a robustness diagnostic rather than an additional
model-selection procedure.
In particular, the largest Walker difference occurs at
$(0.1,10^{-5})$, not at the shared Fixed-HP configuration used in the main
paper.
We therefore do not interpret the grid as evidence that a particular
hyperparameter setting is universally preferable.
Instead, it shows that the Walker result is not confined to a single
optimization point, while also confirming that the magnitude of the effect
depends on the underlying representation-learning configuration.

The task-level decomposition of the shared Fixed-HP result is reported in the
main paper and is not repeated here.
That analysis shows that the aggregate Walker difference is concentrated on
\texttt{stand} and \texttt{run}, whereas \texttt{walk} and \texttt{spin}
operate closer to the reward ceiling.

\section{Cheetah (ExORL) Offline Control}
\label{sec:cheetah}

This section provides the complete numerical results and auxiliary-objective
diagnostics for the ExORL/DMC Cheetah evaluation.
The main paper reports the matched-seed effect and task-level structure;
here we provide additional implementation details, task-level values,
ablation settings, and per-seed results.

\paragraph{Experimental setup.}
We evaluate TwinJEPA on the ExORL/DMC \emph{Cheetah} RND dataset.
Cheetah provides a dense-reward locomotion setting whose dynamics differ
substantially from both maze navigation and bipedal Walker control.
As illustrated in Figure~\ref{fig:cheetah-run-views}, the agent is a planar
articulated system whose performance depends on coordinated joint motion and
sustained velocity rather than explicit goal geometry.

All models are trained for $1\,\mathrm{M}$ gradient steps with batch size
$1024$, discount factor $\gamma=0.98$, and
$\psi_{\mathrm{ortho}}=0.1$.
TwinJEPA uses the shared auxiliary configuration
$\lambda=0.1$,
$\alpha_{\mathrm{gap}}=1.0$, and
$\alpha_{\mathrm{pref}}=0.5$.
Local action-comparison pairs are constructed with
$\tau_{\mathrm{obs}}=0.5$.
TD-JEPA uses the same offline dataset, training schedule, and evaluation
pipeline without the pairwise auxiliary objective.

Zero-shot evaluation covers
\texttt{walk},
\texttt{walk\_backward},
\texttt{run}, and
\texttt{run\_backward}.
Each task is evaluated over 10 episodes using $10\,\mathrm{k}$ inference
samples.
The primary comparison uses 10 matched seeds and reports the mean return
across the four tasks after $1\,\mathrm{M}$ training steps.

\begin{figure}[t]
\centering
\includegraphics[width=0.85\linewidth]
{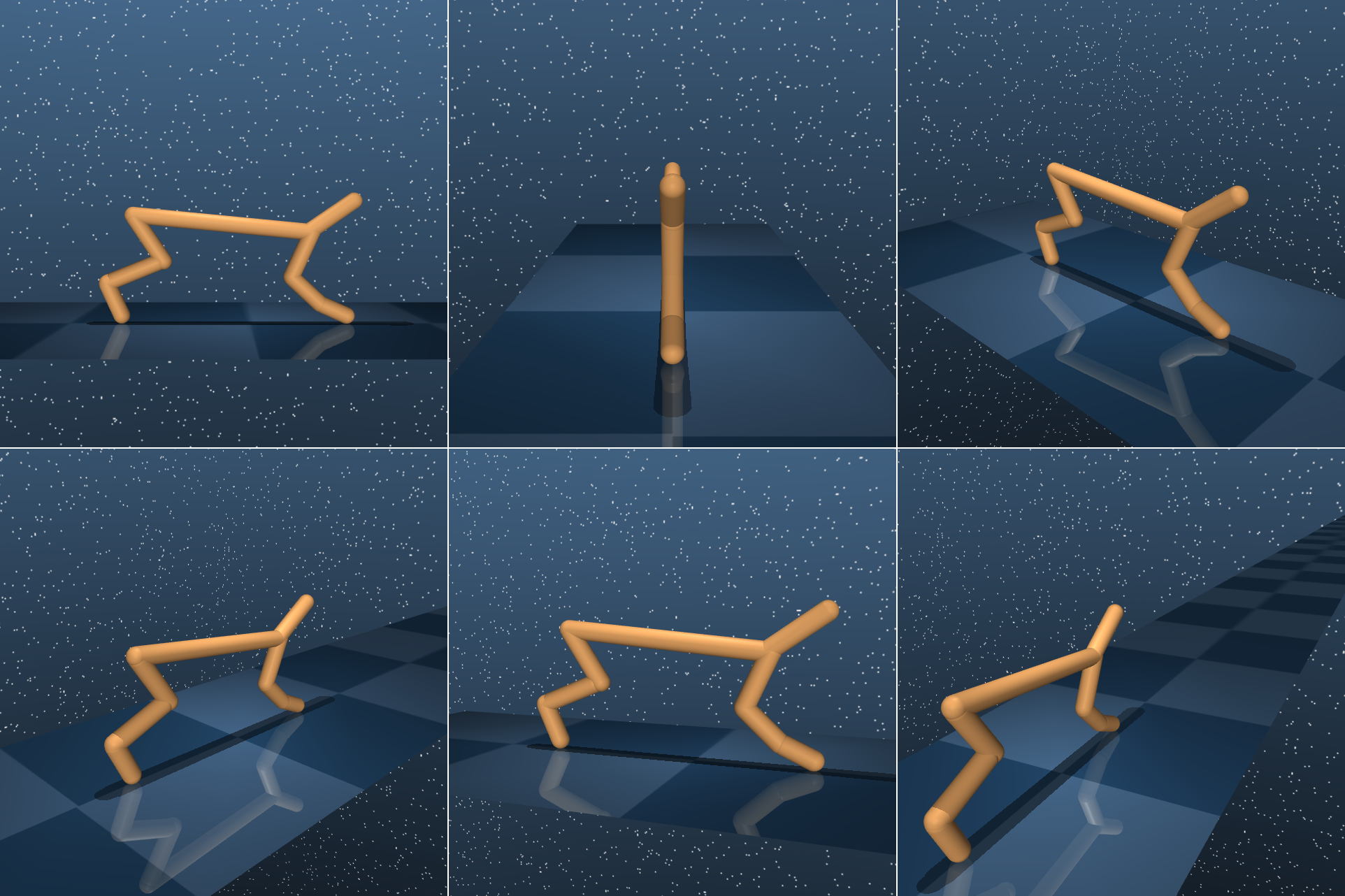}
\caption{
Environment visualizations of the DeepMind Control Suite
\texttt{cheetah-run} environment from six camera viewpoints
(2$\times$3).
Top row: side, front, and three-quarter front views.
Bottom row: three-quarter rear, low-angle, and back views.
}
\label{fig:cheetah-run-views}
\end{figure}

\subsection{Complete Task-Level Results}
\label{app:cheetah-task-results}

Table~\ref{tab:cheetah-per-task} provides the numerical decomposition of the
Cheetah result summarized in the main paper.
The four-task mean increases from $638.0$ to $671.7$
($\Delta{+}{33.7}$), with positive paired differences on 7 of 10 matched
seeds.
The observed across-seed standard deviation is also lower for TwinJEPA
($24.3$ versus $44.2$), although we treat this as a descriptive property of
the evaluated runs rather than as a separate statistical claim.

The task-level differences are heterogeneous.
The largest mean differences occur on the two forward locomotion tasks,
\texttt{walk} and \texttt{run}, whereas
\texttt{walk\_backward} is already close to the apparent return ceiling for
both methods.
This complete numerical breakdown complements the task-level visualization
reported in the main paper.

\begin{table}[t]
\centering
\caption{
Complete task-level zero-shot returns on ExORL/DMC Cheetah under the
Fixed-HP protocol (10 matched seeds, $1\mathrm{M}$ training steps).
Official TD-JEPA and TwinJEPA columns report mean$\pm$standard deviation
across seeds.
The $\Delta$ column reports the corresponding paired seed-wise difference.
The Overall row averages the four evaluation tasks within each seed before
aggregation across seeds.
}
\label{tab:cheetah-per-task}
\setlength{\tabcolsep}{4pt}
\renewcommand{\arraystretch}{1.05}

\begin{tabular}{lccc}
\toprule
Task & Official TD-JEPA & TwinJEPA & $\Delta$ \\
\midrule

\texttt{walk}
& $892.6 \pm 89.7$
& $963.3 \pm 42.9$
& $+70.7 \pm 117.4$
\\

\texttt{walk\_backward}
& $984.3 \pm 0.9$
& $984.5 \pm 0.6$
& $+0.1 \pm 1.0$
\\

\texttt{run}
& $303.7 \pm 88.8$
& $360.4 \pm 58.5$
& $+56.7 \pm 133.6$
\\

\texttt{run\_backward}
& $371.3 \pm 22.8$
& $378.5 \pm 18.3$
& $+7.2 \pm 17.6$
\\

\midrule

Overall
& $638.0 \pm 44.2$
& $\mathbf{671.7 \pm 24.3}$
& $\mathbf{+33.7 \pm 59.1}$
\\

\bottomrule
\end{tabular}

\end{table}

\subsection{Auxiliary-Objective and Weight Sensitivity}
\label{app:cheetah-ablation}

The main paper compares the full, gap-only, and preference-only objectives
across domains.
Here we provide the complete Cheetah ablation results, including additional
values of the auxiliary weight $\lambda$ and the underlying per-seed
returns. The ablation study uses five matched seeds.
Official TD-JEPA and full TwinJEPA are evaluated on this same five-seed
subset, while the remaining variants use identical seeds and evaluation
conditions.
Because this subset differs from the 10-seed primary evaluation, the
corresponding aggregate means should be interpreted within this ablation
experiment and should not be compared directly with the 10-seed values in
Table~\ref{tab:cheetah-per-task}.

\begin{table}[t]
\centering
\caption{
Additional Cheetah auxiliary-objective and loss-weight ablations under the
Fixed-HP protocol (5 matched seeds).
Each entry reports mean$\pm$standard deviation of the four-task return after
$1\mathrm{M}$ training steps.
Gap-only removes preference classification, whereas Pref-only removes
reward-gap regression.
The alternative $\lambda$ settings retain both auxiliary objectives.
}
\label{tab:cheetah-ablation}
\setlength{\tabcolsep}{4pt}
\renewcommand{\arraystretch}{1.05}

\begin{tabular}{lcccc}
\toprule
Variant
& $\lambda$
& $\alpha_{\mathrm{gap}}$
& $\alpha_{\mathrm{pref}}$
& Return @ $1$M
\\
\midrule

Official TD-JEPA
& --
& --
& --
& $610.5 \pm 44.0$
\\

Full TwinJEPA
& $0.1$
& $1.0$
& $0.5$
& $\mathbf{671.4 \pm 32.4}$
\\

\midrule

Gap-only
& $0.1$
& $1.0$
& $0.0$
& $654.5 \pm 34.5$
\\

Pref-only
& $0.1$
& $0.0$
& $0.5$
& $629.8 \pm 37.2$
\\

$\lambda=0.05$
& $0.05$
& $1.0$
& $0.5$
& $634.8 \pm 52.7$
\\

$\lambda=0.5$
& $0.5$
& $1.0$
& $0.5$
& $629.5 \pm 41.7$
\\

\bottomrule
\end{tabular}

\end{table}

Within this five-seed subset, the full objective has the highest mean return
among the evaluated configurations.
Gap-only retains most of the improvement over the matched TD-JEPA baseline,
whereas Pref-only produces a smaller positive difference.
This ordering differs from the AntMaze ablations, where the relative
contributions of the two terms change across domains.
The result therefore supports the interpretation used in the main paper that
reward-gap regression and preference classification provide distinct,
domain-dependent forms of supervision rather than a fixed ranking of
component importance. The auxiliary-weight sweep also shows that the result depends on the strength
of the additional objective.
Both $\lambda=0.05$ and $\lambda=0.5$ produce lower mean returns than the
shared $\lambda=0.1$ configuration on this five-seed subset.
These comparisons do not establish a globally optimal value of $\lambda$;
they instead document the sensitivity of the Cheetah result around the
shared configuration used in the primary experiments.

\begin{table}[t]
\centering
\caption{
Per-seed returns for the Cheetah auxiliary-objective ablations
(5 matched seeds, $1\mathrm{M}$ training steps).
The table exposes the run-to-run variation underlying the aggregate results
in Table~\ref{tab:cheetah-ablation}.
}
\label{tab:cheetah-ablation-seeds}
\setlength{\tabcolsep}{4pt}
\renewcommand{\arraystretch}{1.05}

\begin{tabular}{rcccccc}
\toprule
Seed
& Official
& Full
& Gap-only
& Pref-only
& $\lambda=0.05$
& $\lambda=0.5$
\\
\midrule

3917 & 561.0 & 689.4 & 668.0 & 633.2 & 683.3 & 623.2 \\
3502 & 658.0 & 615.3 & 631.9 & 597.8 & 692.1 & 642.2 \\
8948 & 639.0 & 686.5 & 708.6 & 692.4 & 567.1 & 689.0 \\
9460 & 566.6 & 672.2 & 623.9 & 611.9 & 607.2 & 619.3 \\
4729 & 627.8 & 693.6 & 639.8 & 613.4 & 624.1 & 573.7 \\

\bottomrule
\end{tabular}

\end{table}

%

\begin{figure}[ht]
\centering
\includegraphics[width=\linewidth]
{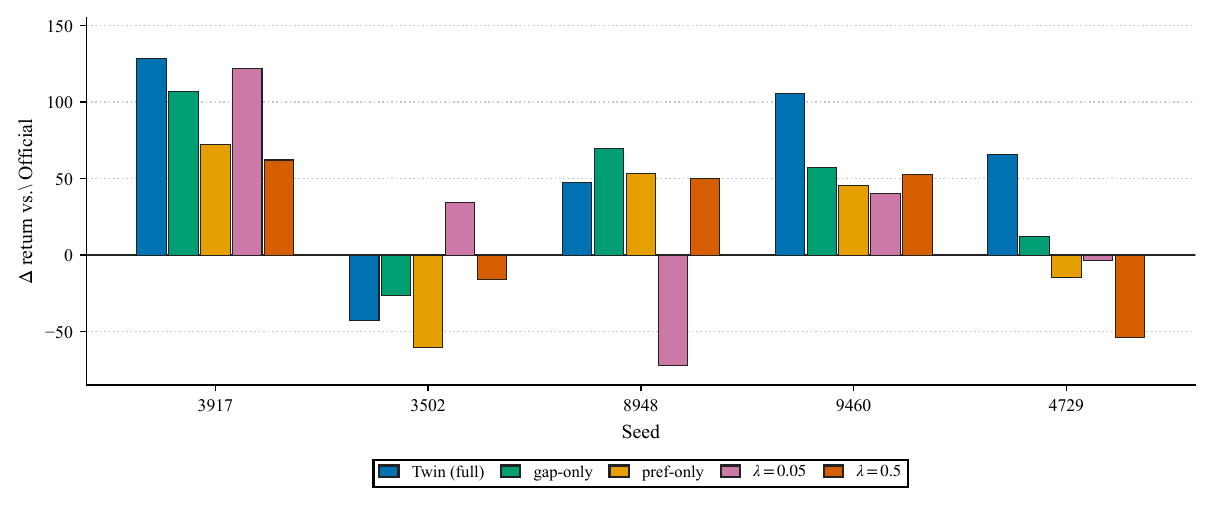}
\caption{
Seed-level return differences for Cheetah auxiliary-objective and
loss-weight ablations under the Fixed-HP protocol.
Each bar reports
$\Delta=\mathrm{Variant}-\mathrm{Official\ TD\text{-}JEPA}$
for one of five matched seeds after $1\mathrm{M}$ training steps.
The figure is intended as a seed-level diagnostic of the aggregate results
in Table~\ref{tab:cheetah-ablation}, rather than as an additional statistical
test.
}
\label{fig:cheetah-ablation-seeds}
\end{figure}

\paragraph{Seed-level behavior.}
Figure~\ref{fig:cheetah-ablation-seeds} and
Table~\ref{tab:cheetah-ablation-seeds} show that the aggregate ordering does
not arise from identical behavior on every seed.
Full TwinJEPA exceeds the matched TD-JEPA baseline on four of the five seeds,
while the component and $\lambda$ variants exhibit different seed-wise
patterns.
These results provide a finer-grained view of the ablation subset and
reinforce the need to interpret the Cheetah ablations in terms of both mean
effects and run-to-run variability.


\section{Quadruped Extended Evaluation}
\label{sec:quadruped}

This section provides an extended analysis of the ExORL/DMC Quadruped
evaluation reported in the main paper.
Among the state-based continuous-control domains considered in our matched
evaluation, Quadruped exhibits the largest domain-level mean increase under
the shared Fixed-HP configuration.
It also differs qualitatively from Walker and Cheetah: the positive effect is
distributed across all four evaluated control objectives rather than being
concentrated on a small subset of tasks.
This makes Quadruped a useful case for examining whether the aggregate gain
reflects a consistent seed-level effect, how it decomposes across individual
tasks, and whether both auxiliary objectives contribute to the observed
improvement.
Accordingly, we report the underlying seed-level returns, exact task-level
values, and an additional auxiliary-weight sensitivity analysis.

\begin{figure}[t]
\centering
\includegraphics[width=0.85\linewidth]
{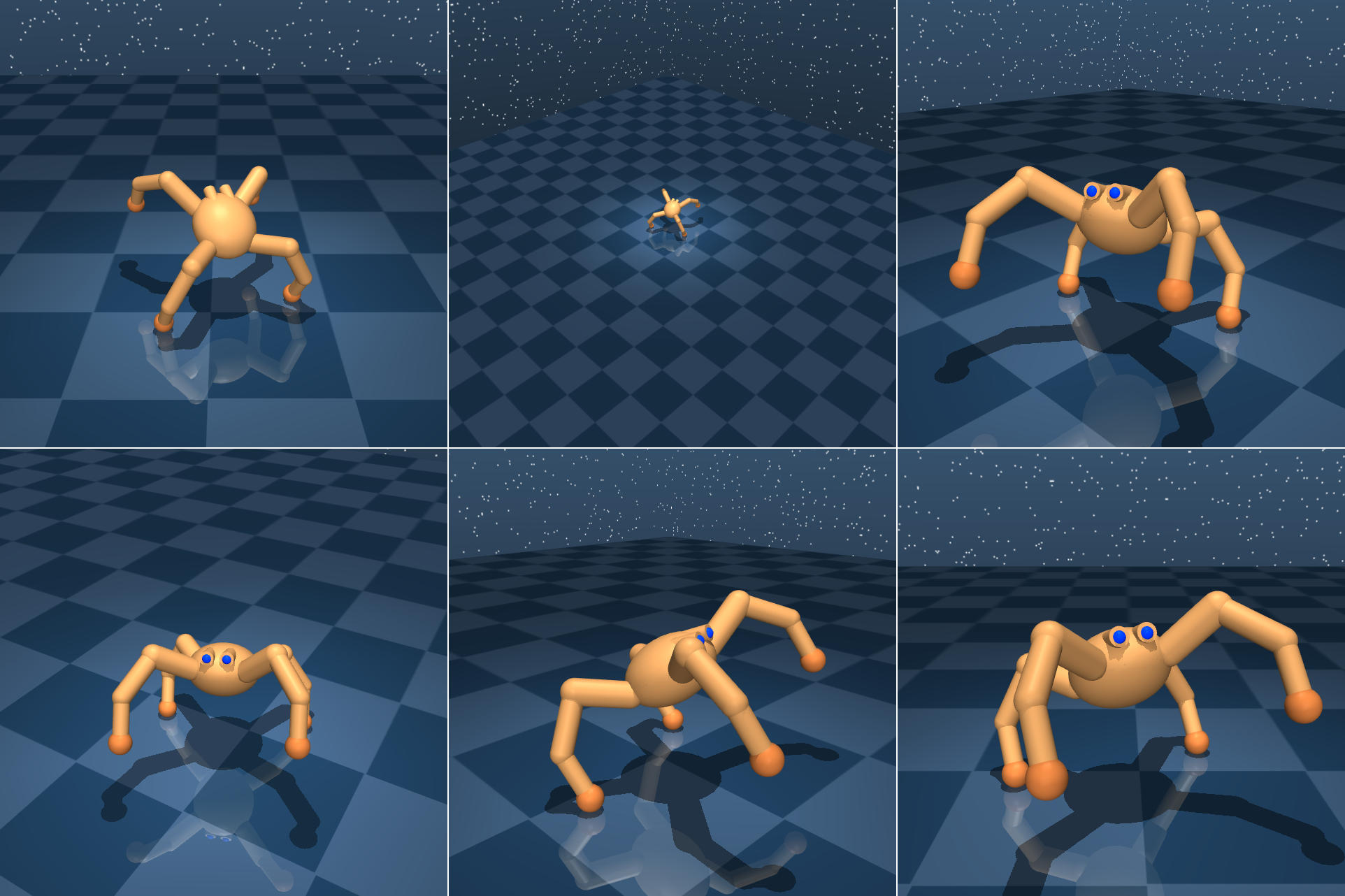}
\caption{
Environment visualizations of the DeepMind Control Suite
\texttt{quadruped-stand} task from six camera viewpoints
(2$\times$3).
The views illustrate the four-legged morphology and spatial configuration
used in the continuous-control evaluation.
}
\label{fig:quadruped-stand-views}
\end{figure}

\paragraph{Experimental setup.}
We evaluate TwinJEPA on the ExORL/DMC Quadruped domain using the
$5\mathrm{k}$-episode RND replay buffer.
Models are trained for $1\mathrm{M}$ gradient steps under the shared Fixed-HP
configuration
$(\psi_{\mathrm{ortho}},\mathrm{lr}_{\psi})=(0.1,10^{-4})$
with five matched seeds.
As in the other state-based comparisons, TD-JEPA and TwinJEPA share the
offline data, backbone, optimization schedule, training budget, task-inference
procedure, random seeds, and evaluation protocol; the additional
action-comparison objective is the systematic difference between the two
methods.

Quadruped observations have a different numerical scale from the other
state-based continuous-control domains.
We therefore use
$\tau_{\mathrm{obs}}=1.0$
for local comparison mining.
This value was selected using a small offline probe intended to avoid an
overly restrictive matching regime rather than by optimizing downstream
evaluation return.
All remaining optimization and evaluation settings follow the shared
protocol.

Zero-shot evaluation averages performance over
\texttt{stand}, \texttt{walk}, \texttt{run}, and \texttt{jump}.
The primary four-task mean increases from
$559.4$ for TD-JEPA to $599.8$ for TwinJEPA
($\Delta{+}{40.4}$).
The paired bootstrap $95\%$ confidence interval reported in the main paper is
$[+20.6,+60.3]$, and all five matched seeds yield positive paired
differences.
Because only five seeds are available, however, the corresponding exact
two-sided sign-flip test attains $p=0.0625$.
We therefore emphasize the effect magnitude, directional consistency, and
task-level structure rather than interpreting the result through a binary
significance threshold.
The analyses below provide the numerical decomposition underlying this
aggregate result.

\subsection{Seed-Level Results}
\label{app:quad-seeds}

To expose the run-to-run variation underlying the aggregate Quadruped result,
Table~\ref{tab:quad-main-seeds} reports the four-task mean return for each of
the five matched seeds.
Because TD-JEPA and TwinJEPA are evaluated with identical seed pairings,
$\Delta$ is computed within each pair.
This paired design reduces variation introduced by comparing independently
sampled seed sets and provides a direct view of how consistently the
TwinJEPA modification changes performance across runs.

\begin{table}[t]
\centering
\caption{
Seed-level zero-shot returns on ExORL/DMC Quadruped under the Fixed-HP
protocol
($\tau_{\mathrm{obs}}=1.0$; five matched seeds; $1\mathrm{M}$ training
steps).
Each row reports the four-task mean for one matched seed, with
$\Delta=\mathrm{TwinJEPA}-\mathrm{Official\ TD\text{-}JEPA}$.
The final row reports mean$\pm$standard deviation across seeds.
}
\label{tab:quad-main-seeds}
\setlength{\tabcolsep}{5pt}
\renewcommand{\arraystretch}{1.05}

\begin{tabular}{rccc}
\toprule
Seed & Official TD-JEPA & TwinJEPA & $\Delta$ \\
\midrule

3502 & 520.8 & 583.1 & $+62.3$ \\
3917 & 612.4 & 634.3 & $+21.9$ \\
4729 & 548.5 & 561.1 & $+12.6$ \\
8948 & 574.5 & 608.4 & $+33.8$ \\
9460 & 540.5 & 612.0 & $+71.6$ \\

\midrule

mean$\pm$std
& $559.4\pm35.4$
& $599.8\pm28.2$
& $+40.4$

\\
\bottomrule
\end{tabular}

\end{table}

TwinJEPA exceeds the matched TD-JEPA baseline on all five seeds, yielding a
mean paired difference of $+40.4$.
The effect is therefore directionally consistent across the available runs,
but its magnitude remains heterogeneous, ranging from $+12.6$ to $+71.6$.
Importantly, the aggregate improvement is not attributable to a single
outlying seed: every matched pair contributes positively, while the spread of
the differences shows that the size of the effect remains sensitive to
stochastic variation. The standard deviation is numerically smaller for TwinJEPA
($28.2$ versus $35.4$), but with only five seeds we do not interpret this
difference as evidence of reduced variance.
Likewise, the $5/5$ positive-seed count should not be treated as a standalone
significance claim.
For five uniformly signed paired differences, the exact two-sided sign-flip
test attains its minimum possible value of $p=0.0625$.
Together with the paired bootstrap interval reported in the main paper, the
seed-level results therefore support a consistent positive direction under
the evaluated protocol while still motivating caution about the precision of
the estimated effect size.

\subsection{Task-Level Numerical Breakdown}
\label{app:quad-task}

The main-paper task-level analysis shows that Quadruped differs from the more
concentrated Walker and Cheetah patterns.
Whereas the larger effects in those domains are concentrated on particular
objectives, all four Quadruped tasks exhibit positive mean differences.
Table~\ref{tab:quad-main-tasks} provides the corresponding numerical
decomposition.

\begin{table}[t]
\centering
\caption{
Task-level zero-shot returns on ExORL/DMC Quadruped under the Fixed-HP
protocol
(five matched seeds; $1\mathrm{M}$ training steps).
Each row reports mean$\pm$standard deviation across seeds for an individual
control objective and the corresponding mean difference
$\Delta=\mathrm{TwinJEPA}-\mathrm{Official\ TD\text{-}JEPA}$.
}
\label{tab:quad-main-tasks}
\setlength{\tabcolsep}{5pt}
\renewcommand{\arraystretch}{1.05}

\begin{tabular}{lccc}
\toprule
Task & Official TD-JEPA & TwinJEPA & $\Delta$ \\
\midrule

\texttt{stand}
& $754.5\pm47.9$
& $804.1\pm35.8$
& $+49.6$
\\

\texttt{walk}
& $462.7\pm77.1$
& $534.8\pm53.1$
& $+72.1$
\\

\texttt{run}
& $384.4\pm23.0$
& $409.2\pm19.1$
& $+24.8$
\\

\texttt{jump}
& $635.9\pm53.2$
& $651.1\pm21.4$
& $+15.2$
\\

\bottomrule
\end{tabular}

\end{table}

The largest mean difference occurs on \texttt{walk} ($+72.1$), followed by
\texttt{stand} ($+49.6$), \texttt{run} ($+24.8$), and \texttt{jump}
($+15.2$).
Thus, the positive domain-level effect is not confined to a single control
objective.
Instead, TwinJEPA exhibits positive mean differences across locomotion,
posture, and jumping behaviors, although the magnitude of the change varies
substantially across tasks.

This broader task-level pattern distinguishes Quadruped from Walker and
Cheetah, where the aggregate gains are concentrated more strongly on selected
objectives.
One possible interpretation is that the locally mined comparisons provide
useful action-discriminative supervision across a wider range of behaviors in
Quadruped.
However, the present experiments do not isolate the cause of this pattern.
Offline coverage, reward structure, baseline performance, control geometry,
and the informativeness of the mined comparison distribution may all
contribute to the observed differences.
We therefore treat the task ordering as descriptive rather than as evidence
for a causal relationship between task difficulty, performance headroom, and
the magnitude of the TwinJEPA effect.

\subsection{Additional Auxiliary-Objective Sensitivity}
\label{app:quadruped-ablation}

The main paper reports the Quadruped Full, Gap-only, and Pref-only
configurations as part of the cross-domain component analysis.
Quadruped is particularly informative for this comparison because the
full objective produces a substantially larger mean return than either
isolated auxiliary term, unlike the component ordering observed on AntMaze
and Cheetah.
For completeness, Table~\ref{tab:quad-ablation} additionally reports the
$\lambda=0.05$ configuration and the differences relative to both TD-JEPA
and the shared full TwinJEPA setting.

\begin{table}[t]
\centering
\caption{
Additional auxiliary-objective results on ExORL/DMC Quadruped under the
Fixed-HP protocol
($\psi_{\mathrm{ortho}}=0.1$,
$\mathrm{lr}_{\psi}=10^{-4}$,
$\tau_{\mathrm{obs}}=1.0$;
five matched seeds;
$1\mathrm{M}$ training steps).
Entries report mean$\pm$standard deviation of the four-task return.
Gap-only removes preference classification, Pref-only removes reward-gap
regression, and the $\lambda=0.05$ variant retains both auxiliary objectives
with a reduced overall auxiliary weight.
}
\label{tab:quad-ablation}
\setlength{\tabcolsep}{4pt}
\renewcommand{\arraystretch}{1.05}

\begin{tabular}{lccc}
\toprule
Variant
& Mean return
& $\Delta$ vs.\ Official
& $\Delta$ vs.\ Full
\\
\midrule

Official TD-JEPA
& $559.4{\pm}35.4$
& --
& $-40.4$
\\

Full TwinJEPA ($\lambda=0.1$)
& $\mathbf{599.8{\pm}28.2}$
& $\mathbf{+40.4}$
& --
\\

Gap-only
& $554.4{\pm}18.2$
& $-5.0$
& $-45.4$
\\

Pref-only
& $548.4{\pm}26.8$
& $-10.9$
& $-51.4$
\\

$\lambda=0.05$
& $569.2{\pm}33.6$
& $+9.8$
& $-30.6$
\\

\bottomrule
\end{tabular}

\end{table}

Within this five-seed Quadruped comparison, the full objective has the highest
mean return among the evaluated variants.
Neither isolated component reproduces the full-objective result:
Gap-only and Pref-only obtain mean returns of $554.4$ and $548.4$,
respectively, compared with $599.8$ for Full TwinJEPA.
Moreover, both isolated components fall slightly below the matched TD-JEPA
baseline.
In this domain, the observed gain therefore appears specifically associated
with the combined comparative objective under the evaluated configuration,
rather than with either reward-gap regression or preference classification
alone. The reduced-weight configuration
$\lambda=0.05$
reaches $569.2$, retaining a positive mean difference of $+9.8$ relative to
TD-JEPA but remaining $30.6$ below the shared $\lambda=0.1$ configuration.
This result indicates that reducing the overall strength of comparative
supervision preserves part, but not all, of the observed Quadruped gain.
Together with the isolated-component results, it suggests that both the
composition and weighting of the auxiliary objective matter in this domain.

These observations should not be generalized into a universal claim of
synergy between reward-gap regression and preference classification.
The cross-domain ablations show that component behavior changes substantially
across AntMaze, Cheetah, and Quadruped: isolated objectives can be competitive
or stronger in other environments, whereas Quadruped is a setting in which
the combined objective is clearly stronger under the evaluated protocol.
The result therefore supports the paper's broader interpretation that the two
signals provide distinct, domain-dependent forms of comparative supervision
rather than a fixed hierarchy of component importance.

Likewise, the $\lambda$ comparison is a local sensitivity analysis rather
than evidence that $\lambda=0.1$ is globally optimal.
Together with the AntMaze and Cheetah sensitivity analyses, it indicates that
the effectiveness of comparative supervision depends not only on whether the
auxiliary signals are present, but also on their relative weighting and their
interaction with the underlying domain.
Overall, the Quadruped results provide the clearest case in our evaluation in
which local action-comparison supervision is positive across matched seeds,
distributed across all evaluated tasks, and strongest when the two auxiliary
signals are used jointly, while the small number of seeds and cross-domain
variation still motivate a deliberately limited interpretation.
\section{D4RL Maze2D Extended Evaluation}
\label{sec:maze2d-evaluation}

\begin{figure}[t]
\centering
\includegraphics[width=0.85\linewidth]
{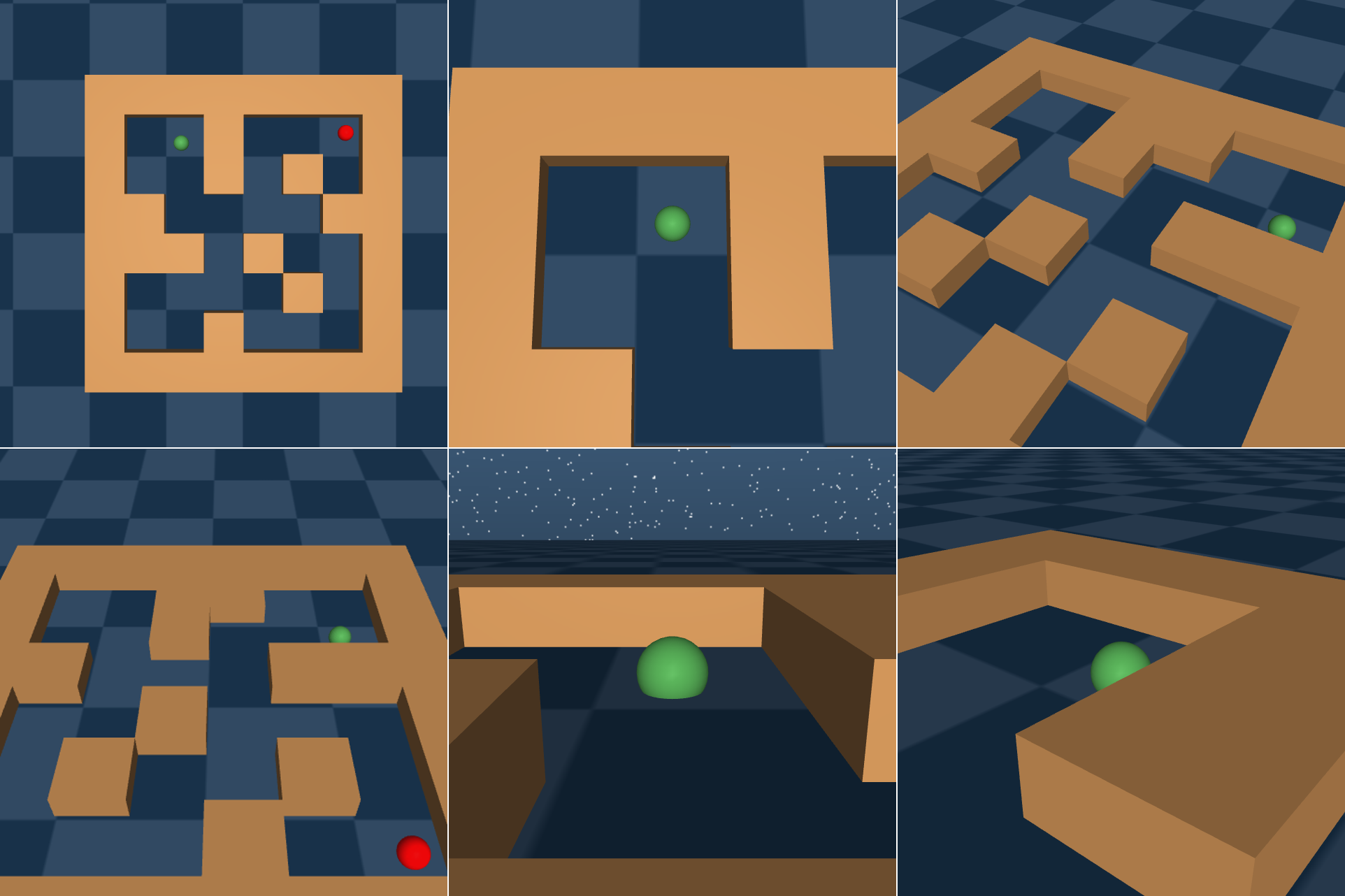}
\caption{
Environment visualizations of the D4RL
\texttt{maze2d-medium-v1} environment from six camera viewpoints
(2$\times$3).
The views illustrate the planar navigation structure used for the
goal-conditioned control evaluation.
}
\label{fig:maze2d-medium-views}
\end{figure}

\paragraph{Maze2D setup.}
We further evaluate TwinJEPA on D4RL
\texttt{maze2d-medium-v1}, a planar point-mass navigation domain that
removes the high-dimensional locomotion dynamics present in AntMaze while
preserving the goal-reaching objective.
This environment provides a complementary evaluation of whether the
comparative supervision introduced by TwinJEPA remains effective when the
control problem is simplified to geometric navigation.

Models are trained for $1\mathrm{M}$ gradient steps using five matched seeds.
The evaluation follows the shared Fixed-HP configuration
$(\psi_{\mathrm{ortho}},\mathrm{lr}_{\psi})=(0.1,10^{-4})$.
Local action-comparison pairs are constructed by bucketing the first two
position coordinates, using
$\tau_{\mathrm{obs}}=0.15$,
$d_b=2$, and
$\epsilon_a=0.01$.
All other training and evaluation procedures follow the common protocol
described in Section~\ref{sec:additional-details}.

\paragraph{Task-level results.}
Figure~\ref{fig:maze2d-fixedhp-pertask} reports the task-level zero-shot
success rates under the shared Fixed-HP configuration.
TwinJEPA improves the five-task mean success rate from
$9.2\%$ to $14.4\%$, corresponding to an absolute difference of
$+5.2$ percentage points.

The improvement is not uniformly distributed across tasks.
T1 and T2 remain at zero success for both methods, whereas TwinJEPA
improves T3, T4, and T5 by $12$, $6$, and $8$ percentage points,
respectively.
Therefore, the aggregate gain at this operating point is mainly contributed
by the subset of tasks where the baseline achieves nonzero performance.
The two unsolved tasks also illustrate that the auxiliary objectives do not
eliminate all difficulties associated with long-horizon goal reaching.

\begin{figure}[ht]
\centering
\includegraphics[width=\linewidth]
{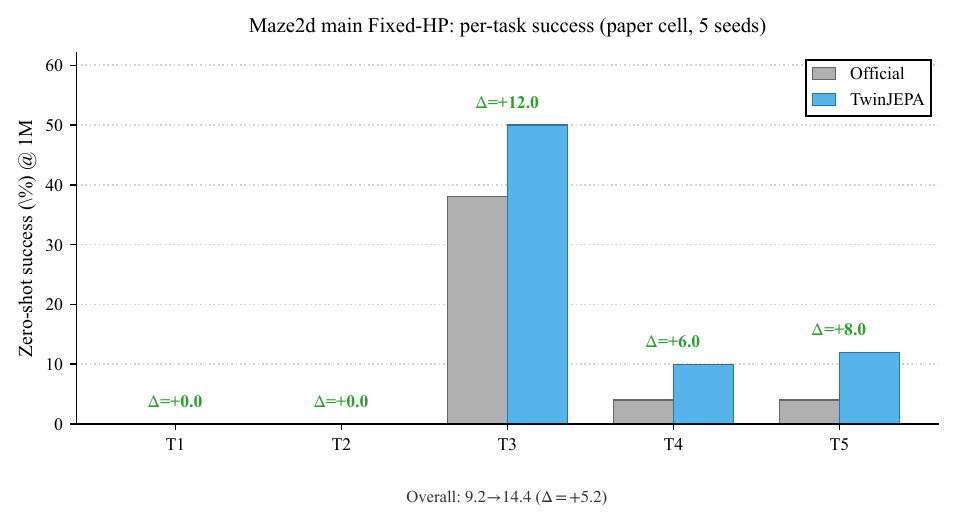}

\caption{
Per-task zero-shot success rates on D4RL
\texttt{maze2d-medium-v1} after $1\mathrm{M}$ training steps under the shared
Fixed-HP configuration
($\psi_{\mathrm{ortho}}=0.1$,
$\mathrm{lr}_{\psi}=10^{-4}$;
five matched seeds).
Bars compare Official TD-JEPA and TwinJEPA, and annotations report
$\Delta=\mathrm{TwinJEPA}-\mathrm{Official\ TD\text{-}JEPA}$
in percentage points.
T1 and T2 remain unsolved by both methods, whereas T3, T4, and T5 show
positive differences of $12$, $6$, and $8$ percentage points.
The five-task mean success rate increases from $9.2\%$ to $14.4\%$
under this shared operating point.
}
\label{fig:maze2d-fixedhp-pertask}

\end{figure}

\paragraph{Interpretation.}
Maze2D provides a complementary but more limited cross-domain evaluation.
Under the shared Fixed-HP configuration, TwinJEPA improves the average
success rate and yields positive differences on the three tasks where
nonzero performance is observed. At the same time, the result should be interpreted together with the
characteristics of this domain.
Unlike AntMaze and continuous-control benchmarks, Maze2D contains simpler
dynamics and a lower-dimensional state space, which changes the type of
information captured by the mined local comparisons.
The current results therefore do not imply that comparative supervision
provides identical benefits across all goal-conditioned navigation settings.

Instead, Maze2D serves as a robustness evaluation showing that the proposed
auxiliary objectives can remain useful beyond high-dimensional locomotion
and complex maze navigation, while also highlighting that their contribution
depends on the interaction between the representation objective and the
specific downstream task.
\section{Pixel-Based ExORL Evaluation}
\label{app:rgb-obs}

The main cross-domain experiments use proprioceptive state observations.
Here we additionally evaluate the same ExORL control domains from rendered
pixel observations to examine whether the TwinJEPA training objective remains
applicable when representations must be learned from visual input.

The state and pixel protocols use the same underlying ExORL RND episodes and
evaluation tasks, but differ in the observations provided to the agent.
Figure~\ref{fig:rgb-vs-state-obs} illustrates this distinction using one
shared trajectory from each domain.
The first three columns show the $64\times64$ RGB frames at
$t{-}2$, $t{-}1$, and $t$, which form the stacked input to the pixel encoder.
The final column shows the proprioceptive state at the same time step, with
dimensions 24, 17, and 78 for Walker, Cheetah, and Quadruped, respectively.

Pixel buffers are generated once by rendering the corresponding offline state
episodes and are then cached for training.
Online rendering is required only during evaluation rollouts.
We therefore treat the pixel experiments as an observation-modality extension
rather than as a new environment or data-collection setting.
State and pixel returns are not compared directly because the two protocols
also differ in encoder architecture, batch size, and training horizon.

\begin{figure}[t]
\centering
\includegraphics[width=\linewidth]
{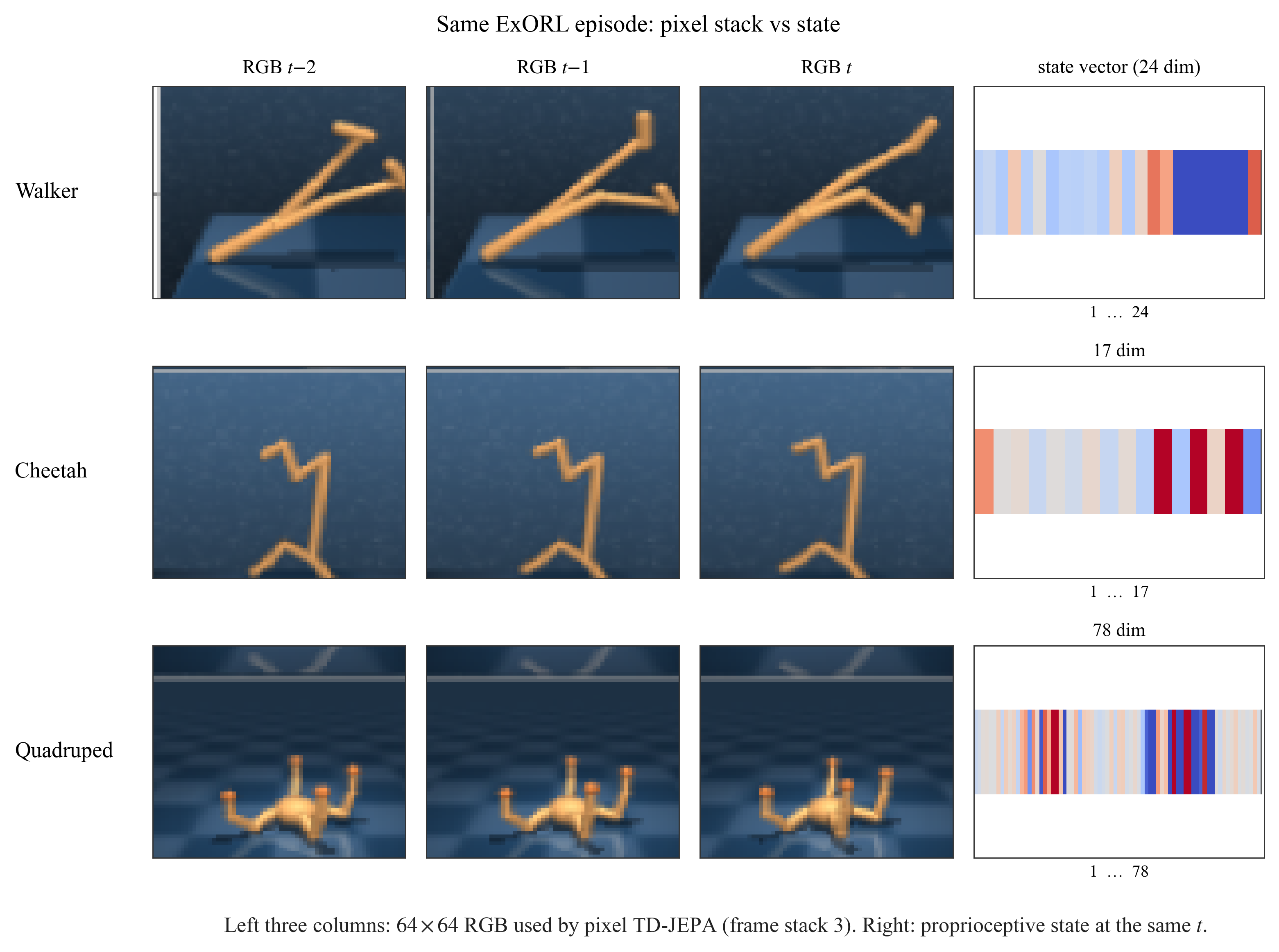}
\caption{
Two observation representations of the same underlying ExORL trajectories.
Each row corresponds to Walker, Cheetah, or Quadruped.
Columns 1--3 show the $64\times64$ RGB frames at
$t{-}2$, $t{-}1$, and $t$ used as the stacked pixel input to the DrQ encoder.
Column 4 shows the proprioceptive state vector at the same time step
(24, 17, and 78 dimensions, respectively).
The figure illustrates the change in observation modality rather than a
direct performance comparison between the state and pixel protocols.
}
\label{fig:rgb-vs-state-obs}
\end{figure}


\paragraph{Pixel ExORL setup.}
We evaluate TD-JEPA and TwinJEPA from pixels on the ExORL Walker, Cheetah,
and Quadruped domains under a shared Fixed-HP protocol.
Both methods use a DrQ encoder with $64\times64$ RGB observations, a frame
stack of three, batch size 512, and
$2\times10^{6}$ gradient steps, with
$\psi_{\mathrm{ortho}}=0.1$ and
$\mathrm{lr}_{\psi}=10^{-4}$.

TwinJEPA uses the same reward-gap regression and preference-classification
objectives as in the state experiments, with $\lambda=0.1$, applied through
the corresponding pixel-specific comparison pipeline.
The auxiliary heads are used only during training and introduce no additional
computation during evaluation.

Evaluation averages 10 episodes per task and then averages task returns
within each domain.
Walker and Quadruped use five matched seeds
($3917,3502,8948,9460,4729$), while Cheetah is available for all ten matched
seeds.
Because the number of seeds differs across domains, the domain-level results
are interpreted separately.
For completeness, we also report an unweighted average over the three domain
means as a descriptive summary rather than as a pooled statistical estimate.

\begin{table}[t]
\centering
\caption{
Pixel-based ExORL performance under the Fixed-HP protocol
($2\times10^{6}$ training steps).
Walker and Quadruped use five matched seeds, while the primary Cheetah result
uses all ten available matched seeds.
The Average row reports the unweighted mean of the three domain-level means.
For Cheetah, we additionally report one five-seed sensitivity subset sampled
with a fixed random seed after excluding $8948$ and $4729$, whose TD-JEPA
\texttt{run} returns are unusually high relative to the remaining runs.
This subset is used only as a seed-composition diagnostic and does not replace
the complete ten-seed comparison.
}
\label{tab:rgb-main}

\begin{tabular}{lcccc}
\toprule
Setting & $n$ & TD-JEPA & TwinJEPA & $\Delta$ \\
\midrule

Walker
& 5
& $716.6 \pm 60.0$
& $\mathbf{741.8 \pm 35.9}$
& $+25.3$
\\

Cheetah
& 10
& $\mathbf{694.8 \pm 23.2}$
& $686.4 \pm 12.1$
& $-8.5$
\\

Quadruped
& 5
& $645.0 \pm 24.6$
& $\mathbf{650.7 \pm 29.0}$
& $+5.6$
\\

\midrule

Unweighted domain avg.
& --
& $685.5$
& $\mathbf{693.0}$
& $+7.5$
\\

\midrule

\multicolumn{5}{l}{\textit{Cheetah seed-set sensitivity}} \\

Random five
& 5
& $685.1 \pm 14.7$
& $\mathbf{689.4 \pm 14.9}$
& $+4.3$
\\

\bottomrule
\end{tabular}
\end{table}

\paragraph{Cheetah seed-set sensitivity.}
The complete ten-seed Cheetah comparison yields a mean difference of
$-8.5$ points.
Inspection of the individual runs shows that TD-JEPA seeds $8948$ and $4729$
have comparatively high \texttt{run} returns of $589$ and $549$,
respectively.

To assess sensitivity to seed composition, we perform an additional
diagnostic after excluding these two seeds.
One five-seed subset is sampled uniformly at random from the remaining eight
seeds using a fixed random seed for reproducibility.
The resulting subset is
$\{2226,3917,4501,6341,9460\}$.
On this subset, TD-JEPA obtains a mean return of $685.1$ and TwinJEPA obtains
$689.4$, corresponding to $\Delta{+}{4.3}$. As this exclusion and resampling analysis is performed after inspecting
the complete runs, we do not treat the random five-seed result as a replacement
for the ten-seed estimate.
Instead, the sign change from $-8.5$ on all ten seeds to $+4.3$ on the
diagnostic subset indicates that the Cheetah pixel result is sensitive to
seed composition.
Accordingly, we do not use Cheetah pixels as evidence for a reliable positive
performance effect.

\paragraph{Task-level diagnostics.}
For Walker and Quadruped, the five-seed domain-level results can be decomposed
directly into their constituent control tasks.
We omit the Cheetah task-level decomposition here because the previously
computed task averages correspond to a different seed subset and are therefore
not consistent with either the complete ten-seed result or the fixed random
five-seed diagnostic reported above.

\begin{table}[t]
\centering
\caption{
Task-level pixel returns for Walker and Quadruped under the same five-seed
pairings used in Table~\ref{tab:rgb-main}.
Cheetah task-level values are omitted because the available decomposition was
computed on a different seed subset from the revised comparisons reported
above.
}
\label{tab:rgb-tasks}

\begin{tabular}{llccc}
\toprule
Domain & Task & TD-JEPA & TwinJEPA & $\Delta$ \\
\midrule

Walker
& stand
& $784.0$
& $\mathbf{887.2}$
& $+103.2$
\\

& walk
& $811.9$
& $\mathbf{815.7}$
& $+3.7$
\\

& spin
& $976.9$
& $\mathbf{978.1}$
& $+1.2$
\\

& run
& $\mathbf{293.4}$
& $286.4$
& $-7.1$
\\

\midrule

Quadruped
& run
& $453.8$
& $\mathbf{472.3}$
& $+18.5$
\\

& walk
& $461.6$
& $\mathbf{470.7}$
& $+9.1$
\\

& stand
& $904.6$
& $\mathbf{917.4}$
& $+12.9$
\\

& jump
& $\mathbf{760.2}$
& $742.3$
& $-17.9$
\\

\bottomrule
\end{tabular}
\end{table}

\paragraph{Task-level behavior.}
On Walker, the positive domain-level difference is concentrated primarily on
\texttt{stand}, which increases from $784.0$ to $887.2$.
The differences on \texttt{walk} and \texttt{spin} are small, while
\texttt{run} decreases by $7.1$ points.
The Walker pixel result therefore reflects a concentrated task-level effect
rather than a uniform improvement across all four control objectives.

Quadruped shows a smaller positive domain-level difference.
TwinJEPA has higher mean returns on \texttt{run}, \texttt{walk}, and
\texttt{stand}, with differences of $+18.5$, $+9.1$, and $+12.9$,
respectively, while \texttt{jump} decreases by $17.9$ points.
The task-level pattern is therefore mixed despite the positive domain mean.

\paragraph{Overall interpretation.}
The pixel experiments are intended primarily as an observation-modality
extension of the state-based evaluation rather than as an additional claim
of uniform performance improvement.
They show that the same comparative representation-learning objective can be
applied when observations are provided as cached RGB frame stacks rather than
proprioceptive state vectors.

Across the three primary domain comparisons, the unweighted domain mean
changes from $685.5$ to $693.0$, corresponding to a descriptive difference
of $+7.5$ points.
This aggregate should be interpreted cautiously because the domains use
different numbers of matched seeds and the observed effects are
heterogeneous.

Walker shows the largest positive mean difference, concentrated primarily on
\texttt{stand}.
Quadruped shows a smaller positive mean difference with mixed task-level
effects.
Cheetah differs from both domains, with the complete ten-seed comparison
favoring TD-JEPA and the fixed random five-seed diagnostic favoring TwinJEPA.
The Cheetah result therefore illustrates sensitivity to seed composition
rather than a stable directional effect.

Taken together, these experiments support the applicability of the TwinJEPA
training objective to visual observations, but they do not establish
modality-independent or uniformly positive performance gains.
The pixel results also remain separate from the state-based comparisons
because the encoder architecture, batch size, and training horizon differ
between the two protocols.

\section{Theoretical Properties of TwinJEPA}
\label{sec:theory-twin}

We analyze TwinJEPA as an additive extension of the TD-JEPA representation
objective and characterize the supervision introduced by reward-gap
regression and preference classification.
The analysis concerns the fixed distribution of local action-comparison
pairs constructed from the offline dataset.
It establishes four main properties and an affine-head specialization.

\begin{enumerate}

\item \textbf{Target consistency.}
The reward-gap and preference targets are derived from the same selected
signed reward contrast.
Reward-gap regression uses its magnitude, whereas preference classification
uses its induced ordering.

\item \textbf{Exact reduction to TD-JEPA.}
When the auxiliary contribution is disabled, the TwinJEPA objective reduces
exactly to the original TD-JEPA objective.
Any theoretical property established for TD-JEPA under its original
assumptions therefore applies to this reduced case.

\item \textbf{Reward-gap decodability.}
For a fixed representation, the minimum achievable squared reward-gap risk is
the expected conditional variance of the reward-gap magnitude given that
representation.
Reducing this risk therefore requires the representation to preserve
information predictive of outcome-contrast magnitude under the mined pair
distribution.

\item \textbf{Local ranking control.}
On the mined pair distribution, the preference cross-entropy upper-bounds the
probability that the learned scorer reverses the reward-induced ordering of a
retained action pair, up to the factor $1/\log 2$.

\end{enumerate}

These statements characterize the representation-level information promoted
by the auxiliary objectives.
They do not imply that TwinJEPA improves every downstream task, that the two
auxiliary terms are universally complementary, or that the theoretical
guarantees of TD-JEPA continue to hold unchanged when the auxiliary weights
are nonzero.

\paragraph{Setup.}
Let
$
\phi:\mathcal{O}\rightarrow\mathbb{R}^{d_{\phi}}
$
denote the observation encoder, and let
$\mathcal{L}_{\mathrm{base}}$ denote the original TD-JEPA representation
objective used in the main method.
TwinJEPA optimizes
\begin{equation}
\label{eq:twin-total}
\mathcal{L}_{\mathrm{Twin}}
=
\mathcal{L}_{\mathrm{base}}
+
\lambda
\left(
\alpha_{\mathrm{gap}}\mathcal{L}_{\mathrm{gap}}
+
\alpha_{\mathrm{pref}}\mathcal{L}_{\mathrm{pref}}
\right),
\end{equation}
where
$\lambda,\alpha_{\mathrm{gap}},\alpha_{\mathrm{pref}}\geq 0$.

A retained local action-comparison sample is written as
$
\zeta
=
(o_{\mathrm{anc}},a_A,a_B,\Delta r,y),
$
where $o_{\mathrm{anc}}$ is the anchor observation used by the auxiliary
heads, $a_A$ and $a_B$ are the compared actions,
$\Delta r\in\mathbb{R}_{\geq 0}$ is the selected reward-gap magnitude, and
$y\in\{0,1\}$ is the associated preference label.
Let $\mathcal{P}_{\mathrm{pair}}$ denote the resulting pair distribution and
define
$
Z=\phi(o_{\mathrm{anc}}).
$

The analysis below treats $\mathcal{P}_{\mathrm{pair}}$ as fixed.
Thus, it characterizes the supervision induced by a given pair-mining
procedure rather than the statistical quality of the mining procedure itself.

Under the implemented auxiliary-head interfaces, the reward-gap predictor
uses $Z$, whereas the preference predictor additionally receives the two
compared actions.
The corresponding population objectives are
\begin{align}
\label{eq:twin-gap}
\mathcal{L}_{\mathrm{gap}}
&=
\mathbb{E}_{\zeta\sim\mathcal{P}_{\mathrm{pair}}}
\left[
\left(
g_{\theta}(Z)-\Delta r
\right)^2
\right],
\\
\label{eq:twin-pref}
\mathcal{L}_{\mathrm{pref}}
&=
\mathbb{E}_{\zeta\sim\mathcal{P}_{\mathrm{pair}}}
\left[
-\log
p_{\theta}
\left(
y\mid Z,a_A,a_B
\right)
\right].
\end{align}

\begin{assumption}[Goal-conditioned reward family]
\label{ass:goal-rewards}

Let $\{r_k\}_{k=1}^{K}$ denote the finite family of goal-conditioned reward
functions used by the relabeling rule that constructs the comparison targets.
Consider two transitions $A$ and $B$ with next physical states
$x'_A$ and $x'_B$ and actions $a_A$ and $a_B$.
For each goal $k$, define the signed reward contrast
\begin{equation}
\label{eq:goal-gap}
\delta r_k
=
r_k(x'_A,a_A)
-
r_k(x'_B,a_B).
\end{equation}

The selected goal index is
\begin{equation}
\label{eq:selected-goal}
k^{\star}
=
\min
\left\{
k\in\{1,\ldots,K\}
:
|\delta r_k|
=
\max_{j\in\{1,\ldots,K\}}
|\delta r_j|
\right\},
\end{equation}
where the minimum index provides a deterministic tie-breaking convention.

Define the selected signed contrast
$
\delta r^\star
=
\delta r_{k^\star}.
$
Consistent with the main method, the pair targets are
\begin{equation}
\label{eq:pair-labels}
\Delta r
=
|\delta r^\star|,
\qquad
y
=
\mathbf{1}
\left\{
\delta r^\star<0
\right\}.
\end{equation}
Thus, $\Delta r$ records the magnitude of the selected outcome contrast,
while $y=0$ denotes that transition $A$ has the larger selected reward and
$y=1$ denotes that transition $B$ has the larger selected reward.

\end{assumption}

\begin{assumption}[Validity of retained local comparisons]
\label{ass:pair-validity}

Every retained pair satisfies the local observation-matching rule used by the
pair miner together with
$
\lVert a_A-a_B\rVert_2
\geq
\varepsilon_a,
\qquad
\Delta r
\geq
\varepsilon_r,
$
for $\varepsilon_a>0$ and $\varepsilon_r>0$.

Let
$\mathcal{M}_{\tau}(o_A,o_B)=1$
denote that the two observations satisfy the corresponding local matching
rule.
The retained pairs therefore satisfy
$
\mathcal{M}_{\tau}(o_A,o_B)=1.
$

For the AntMaze state experiments, this rule is implemented by
$
b(o)
=
\operatorname{round}
\left(
\frac{q(o)}{\tau_{\mathrm{obs}}}
\right),
$
where $q(o)$ contains the first $d_b=2$ positional coordinates and a pair is
eligible only when
$b(o_A)=b(o_B)$.
The reported AntMaze configuration uses
$\tau_{\mathrm{obs}}=0.3$,
$\varepsilon_a=0.01$, and
$\varepsilon_r=10^{-6}$.

\end{assumption}

\begin{proposition}[Consistency of the auxiliary targets]
\label{prop:target-consistency}

Under Assumptions~\ref{ass:goal-rewards} and~\ref{ass:pair-validity},
\begin{equation}
\label{eq:target-consistency}
\begin{aligned}
\Delta r
&=
|\delta r^\star|
>0,
\\
y=0
&\iff
\delta r^\star>0
\iff
r_{k^{\star}}(x'_A,a_A)
>
r_{k^{\star}}(x'_B,a_B),
\\
y=1
&\iff
\delta r^\star<0
\iff
r_{k^{\star}}(x'_A,a_A)
<
r_{k^{\star}}(x'_B,a_B).
\end{aligned}
\end{equation}

Moreover, $k^{\star}$ maximizes the absolute one-step reward contrast over the
candidate reward family.
Thus, reward-gap regression and preference classification are derived from
the same selected signed contrast $\delta r^\star$: $\Delta r$ captures its
magnitude, whereas $y$ captures its ordering.

\end{proposition}

\begin{proof}[Proof sketch]

By Assumption~\ref{ass:goal-rewards},
$
\Delta r
=
|\delta r^\star|
$
and
$
y
=
\mathbf{1}\{\delta r^\star<0\}.
$
Assumption~\ref{ass:pair-validity} gives
$\Delta r\geq\varepsilon_r>0$, excluding ties.

If $\delta r^\star>0$, transition $A$ has the larger selected reward and
Eq.~\ref{eq:pair-labels} gives $y=0$.
If $\delta r^\star<0$, transition $B$ has the larger selected reward and
Eq.~\ref{eq:pair-labels} gives $y=1$.
The maximal-contrast property follows directly from
Eq.~\ref{eq:selected-goal}.

\end{proof}

\begin{proposition}[Exact reduction to TD-JEPA]
\label{prop:reduction}

If $\lambda=0$, or if
$\alpha_{\mathrm{gap}}=\alpha_{\mathrm{pref}}=0$, then
\begin{equation}
\label{eq:twin-reduction}
\mathcal{L}_{\mathrm{Twin}}
=
\mathcal{L}_{\mathrm{base}}.
\end{equation}
Therefore, under either condition, TwinJEPA reduces exactly to the original
TD-JEPA representation objective.

\end{proposition}

\begin{proof}[Proof sketch]

If $\lambda=0$, the auxiliary term in
Eq.~\ref{eq:twin-total} vanishes.
If
$\alpha_{\mathrm{gap}}=\alpha_{\mathrm{pref}}=0$,
the weighted auxiliary loss is zero for any $\lambda$.
Eq.~\ref{eq:twin-reduction} follows immediately.

\end{proof}

Consequently, any TD-JEPA result that holds under the assumptions of the
original objective also applies to this exact-reduction case.
In particular, the theoretical results established for TD-JEPA in its
idealized analysis apply when the TwinJEPA auxiliary contribution is
disabled~\citep{bagatella2025td}.
This statement does not extend those guarantees to nonzero auxiliary weights.

\begin{proposition}[Reward-gap magnitude Bayes risk]
\label{prop:gap-information}

Let $Z=\phi(o_{\mathrm{anc}})$ and assume
$\mathbb{E}[(\Delta r)^2]<\infty$.
Among all measurable predictors
$g:\mathbb{R}^{d_{\phi}}\rightarrow\mathbb{R}$,
the population reward-gap risk is minimized by
\begin{equation}
\label{eq:bayes-gap-predictor}
g^{\star}(Z)
=
\mathbb{E}
\left[
\Delta r\mid Z
\right].
\end{equation}

Furthermore,
\begin{equation}
\label{eq:minimum-gap-risk}
\inf_g
\mathcal{L}_{\mathrm{gap}}(g)
=
\mathbb{E}
\left[
\operatorname{Var}
\left(
\Delta r\mid Z
\right)
\right].
\end{equation}

Thus, the minimum reward-gap risk is zero if and only if
$\Delta r$ is almost surely determined by $Z$ under
$\mathcal{P}_{\mathrm{pair}}$.

\end{proposition}

\begin{proof}

Define
$
m(Z)
=
\mathbb{E}
\left[
\Delta r\mid Z
\right].
$
For any measurable predictor $g$,
$
g(Z)-\Delta r
=
g(Z)-m(Z)
+
m(Z)-\Delta r.
$
Expanding the squared error and taking expectations gives
\begin{align}
\mathcal{L}_{\mathrm{gap}}(g)
&=
\mathbb{E}
\left[
\left(
g(Z)-m(Z)
\right)^2
\right]
+
\mathbb{E}
\left[
\left(
m(Z)-\Delta r
\right)^2
\right]
\nonumber\\
&\quad
+
2\,
\mathbb{E}
\left[
\left(
g(Z)-m(Z)
\right)
\left(
m(Z)-\Delta r
\right)
\right].
\label{eq:gap-decomposition-expanded}
\end{align}

The cross term is zero because
$
\mathbb{E}
\left[
m(Z)-\Delta r
\mid Z
\right]
=
0.
$
Moreover,
$
\mathbb{E}
\left[
\left(
m(Z)-\Delta r
\right)^2
\right]
=
\mathbb{E}
\left[
\operatorname{Var}
\left(
\Delta r\mid Z
\right)
\right].
$

Therefore,
\begin{equation}
\label{eq:gap-risk-decomposition}
\mathcal{L}_{\mathrm{gap}}(g)
=
\mathbb{E}
\left[
\left(
g(Z)-m(Z)
\right)^2
\right]
+
\mathbb{E}
\left[
\operatorname{Var}
\left(
\Delta r\mid Z
\right)
\right].
\end{equation}

The first term is nonnegative and vanishes when
$g(Z)=m(Z)$ almost surely, proving
Eq.~\ref{eq:bayes-gap-predictor} and
\ref{eq:minimum-gap-risk}.
The minimum risk equals zero precisely when
$\operatorname{Var}(\Delta r\mid Z)=0$ almost surely.

\end{proof}

Proposition~\ref{prop:gap-information} provides a representation-level
interpretation of the implemented state-only reward-gap head.
For a fixed pair distribution, variation in reward-gap magnitude that cannot
be predicted from $Z$ contributes to the irreducible squared-error term.
Lowering this risk therefore requires $Z$ to preserve information predictive
of the local outcome-contrast magnitudes represented in
$\mathcal{P}_{\mathrm{pair}}$.
The proposition does not require all reward-relevant information in the
environment to be encoded.

\begin{proposition}[Affine reward-gap specialization]
\label{prop:linear-gap}

Let
$
\overline{Z}
=
\begin{bmatrix}
Z\\
1
\end{bmatrix},
$
and consider the affine predictor
$
g_w(Z)
=
w^{\top}\overline{Z}.
$
Assume
$\mathbb{E}[\|\overline{Z}\|_2^2]<\infty$
and that
$
\Sigma_{\overline{Z}}
=
\mathbb{E}
\left[
\overline{Z}\,
\overline{Z}^{\top}
\right]
$
is invertible.
Then the population least-squares minimizer is
\begin{equation}
\label{eq:linear-gap-solution}
w^{\star}
=
\Sigma_{\overline{Z}}^{-1}
\mathbb{E}
\left[
\overline{Z}\,\Delta r
\right],
\end{equation}
and the corresponding minimum risk is
\begin{equation}
\label{eq:linear-gap-risk}
\min_w
\mathcal{L}_{\mathrm{gap}}(w)
=
\mathbb{E}
\left[
(\Delta r)^2
\right]
-
\mathbb{E}
\left[
\overline{Z}\,\Delta r
\right]^{\top}
\Sigma_{\overline{Z}}^{-1}
\mathbb{E}
\left[
\overline{Z}\,\Delta r
\right].
\end{equation}

In particular, this affine-head risk is zero if and only if
$\Delta r=w^{\top}\overline{Z}$ almost surely for some $w$.

\end{proposition}

\begin{proof}[Proof sketch]

Expanding the population objective gives
$
\mathcal{L}_{\mathrm{gap}}(w)
=
\mathbb{E}
\left[
(\Delta r)^2
\right]
-
2w^{\top}
\mathbb{E}
\left[
\overline{Z}\,\Delta r
\right]
+
w^{\top}
\Sigma_{\overline{Z}}
w.
$
Differentiating with respect to $w$ gives
$
\Sigma_{\overline{Z}}w^{\star}
=
\mathbb{E}
\left[
\overline{Z}\,\Delta r
\right].
$
Invertibility of $\Sigma_{\overline{Z}}$ yields
Eq.~\ref{eq:linear-gap-solution}.
Substitution gives Eq.~\ref{eq:linear-gap-risk}.
The minimum risk is zero exactly when the affine least-squares residual
vanishes almost surely.

\end{proof}

\begin{proposition}[Preference classification and local ranking]
\label{prop:preference-ranking}

Suppose the preference classifier is parameterized by a scalar logit
$f_{\theta}$ satisfying
\begin{equation}
\label{eq:preference-logit}
p_{\theta}
\left(
y=0\mid Z,a_A,a_B
\right)
=
\sigma
\left(
f_{\theta}(Z,a_A,a_B)
\right),
\end{equation}
where
$\sigma(t)=(1+\exp(-t))^{-1}$.
Define
$
t=1-2y\in\{-1,1\},
$
so that, by Proposition~\ref{prop:target-consistency},
$
t=\operatorname{sign}(\delta r^\star).
$
The local ranking error on the mined pair distribution is
\begin{equation}
\label{eq:ranking-error}
\mathcal{E}_{\mathrm{rank}}(\theta)
=
\mathbb{P}_{\zeta\sim\mathcal{P}_{\mathrm{pair}}}
\left[
t\,
f_{\theta}(Z,a_A,a_B)
\leq 0
\right].
\end{equation}

Then
\begin{equation}
\label{eq:ranking-error-bound}
\mathcal{E}_{\mathrm{rank}}(\theta)
\leq
\frac{
\mathcal{L}_{\mathrm{pref}}(\theta)
}{
\log 2
}.
\end{equation}

Consequently,
$\mathcal{L}_{\mathrm{pref}}(\theta)\rightarrow0$
implies
$\mathcal{E}_{\mathrm{rank}}(\theta)\rightarrow0$
on $\mathcal{P}_{\mathrm{pair}}$.

\end{proposition}

\begin{proof}

Under Eq.~\ref{eq:preference-logit}, the per-sample binary
cross-entropy is
\begin{equation}
\label{eq:preference-logistic-loss}
\ell_{\mathrm{pref}}
=
\log
\left(
1+
\exp
\left(
-t f_{\theta}(Z,a_A,a_B)
\right)
\right).
\end{equation}

A ranking error occurs when
$
t f_{\theta}(Z,a_A,a_B)\leq0.
$
On this event,
$
\exp
\left(
-t f_{\theta}(Z,a_A,a_B)
\right)
\geq1,
$
and therefore
$
\ell_{\mathrm{pref}}
\geq
\log2.
$

Hence
$
\ell_{\mathrm{pref}}
\geq
\log2\,
\mathbf{1}
\left\{
t f_{\theta}(Z,a_A,a_B)\leq0
\right\}.
$
Taking expectations over $\mathcal{P}_{\mathrm{pair}}$ gives
$
\mathcal{L}_{\mathrm{pref}}(\theta)
\geq
\log2\,
\mathcal{E}_{\mathrm{rank}}(\theta),
$
which proves Eq.~\ref{eq:ranking-error-bound}.

\end{proof}

Proposition~\ref{prop:preference-ranking} gives a direct interpretation of
the preference-classification objective on the retained comparison
distribution.
Low cross-entropy constrains the probability that the learned scorer reverses
the reward-induced ordering of a sampled action pair.
The result concerns ranking on $\mathcal{P}_{\mathrm{pair}}$ and does not
imply correct ordering for arbitrary states, actions, or goals outside that
distribution.

\begin{remark}[Interpretation and scope]
\label{rem:theory-scope}

Propositions~\ref{prop:target-consistency}--\ref{prop:preference-ranking}
characterize the additional supervision introduced by TwinJEPA.
Reward-gap regression promotes decodability of the magnitudes of local
outcome contrasts in the mined comparison distribution, whereas preference
classification directly penalizes incorrect ordering of the compared
actions.
The two objectives are derived from the same selected signed contrast
$\delta r^\star$ but use complementary information:
$\Delta r=|\delta r^\star|$ captures magnitude, while $y$ captures ordering.
They are therefore related but not mathematically interchangeable.
The theory motivates these distinct roles without implying that their
combination must outperform either component in every domain.
This distinction is consistent with the component ablations reported in the
main paper.

Proposition~\ref{prop:reduction} transfers the theoretical properties of
TD-JEPA only in the exact-reduction case where the auxiliary contribution is
zero.
For nonzero auxiliary weights, TwinJEPA optimizes a different joint
objective, and the analysis above does not establish preservation of
TD-JEPA's original downstream guarantees.

The conclusions are also conditional on the pair distribution
$\mathcal{P}_{\mathrm{pair}}$.
Changing the local matching rule, reward relabeling procedure, action
threshold, or reward-gap threshold changes this distribution and therefore
the supervision characterized by the propositions.
The theory does not establish that any particular pair-mining rule is
optimal.

Accordingly, the practical effect of TwinJEPA can depend on offline-data
coverage, the informativeness of the mined comparisons, their alignment with
the downstream objective, and the relative weighting of the auxiliary
losses.
These factors are examined empirically through the task-level, component,
weight-sensitivity, and cross-domain analyses rather than treated as
consequences of the theoretical results above.

\end{remark}

\section{Summary of Additional Results}
\label{sec:additional-results-summary}

The additional experiments complement the main evaluation along several
dimensions, including task-level heterogeneity, auxiliary-objective
contributions, cross-domain transfer, optimization sensitivity, and
observation modality.
Together with the theoretical analysis in
Section~\ref{sec:theory-twin}, these results support a deliberately limited
interpretation.
TwinJEPA introduces supervision that promotes reward-gap decodability and
reward-consistent local action ranking on the mined comparison distribution,
while the downstream effect of this supervision varies with the task, domain,
optimization configuration, and comparison distribution.

\paragraph{The observed effects are heterogeneous across tasks.}

The AntMaze decomposition shows that TwinJEPA does not produce a uniform
increase across all evaluation goals.
Under the shared Fixed-HP protocol, the AntMaze-4 mean changes from
$50.1\%$ to $52.4\%$ ($\Delta{+}{2.3}$), with non-negative mean
differences on all four evaluated domains.
At the task level, however, positive and negative differences coexist.
For example, medium-navigate Task~3 improves by $17$ percentage points,
whereas several other goals change little or decrease.

The same type of heterogeneity appears in the continuous-control domains.
On Walker, the larger positive differences occur on \texttt{stand} and
\texttt{run}, while \texttt{walk} and \texttt{spin} change little.
On Cheetah, the largest differences occur on the forward
\texttt{walk} and \texttt{run} tasks, while
\texttt{walk\_backward} is nearly unchanged.
Quadruped differs in that all four evaluated objectives have positive mean
differences, although their magnitudes vary substantially. These patterns are descriptively consistent with larger gains occurring on
some tasks that retain substantial performance headroom, but the experiments
do not establish headroom as the causal mechanism.
Offline coverage, reward structure, optimization dynamics, and the
informativeness of the mined comparison pairs may all contribute to the
observed task-level differences.

\paragraph{The auxiliary objectives provide distinct and domain-dependent supervision.}

The component ablations show that reward-gap regression and preference
classification are not interchangeable.
On AntMaze medium-navigate, preference-only has the highest mean among the
evaluated component variants, whereas on medium-stitch gap-only performs
better.
On Cheetah, the five-seed ablation subset gives mean returns of
$671.4$ for Full TwinJEPA, $654.5$ for Gap-only, and $629.8$ for
Pref-only.
On Quadruped, Full TwinJEPA reaches $599.8$, compared with $554.4$ for
Gap-only and $548.4$ for Pref-only.

These results do not imply a universal synergy between the two losses.
Instead, they show that their relative contributions depend on the domain and
the evaluated configuration.
This empirical behavior is consistent with the theoretical distinction in
Section~\ref{sec:theory-twin}.
Reward-gap regression uses the magnitude of the selected reward contrast,
whereas preference classification supervises only its induced ordering.
The two targets are constructed from the same local comparison but emphasize
different information.

The auxiliary weight also affects the observed result.
AntMaze, Cheetah, and Quadruped sensitivity analyses show that changing
$\lambda$ can alter the magnitude of the downstream effect.
We therefore retain $\lambda=0.1$ as a shared evaluation configuration
rather than interpreting it as a per-domain optimum.

\paragraph{The effect transfers beyond the primary AntMaze benchmark.}

The continuous-control evaluations provide evidence that the TwinJEPA
objective is not restricted to the AntMaze setting.
Under the shared Fixed-HP protocol, Walker changes from $774.3$ to $782.7$
($\Delta{+}{8.4}$), Cheetah from $638.0$ to $671.7$
($\Delta{+}{33.7}$), and Quadruped from $559.4$ to $599.8$
($\Delta{+}{40.4}$).

These mean effects should be interpreted together with their matched-seed
uncertainty.
The paired analyses in the main paper show substantial uncertainty for
Walker and Cheetah, while Quadruped has positive differences on all five
matched seeds.
Because the Quadruped comparison contains only five seeds, the exact
two-sided test has limited resolution, and we therefore emphasize effect
magnitude and directional consistency rather than a binary significance
claim.

The cross-domain results are consequently most useful as evidence that the
same comparative objective can be applied across different control
morphologies and reward structures, while the size and consistency of the
result remain domain dependent.

\paragraph{Maze2D provides a boundary-case navigation evaluation.}

Maze2D removes the high-dimensional locomotion dynamics of AntMaze and
provides a simpler planar navigation setting.
Under the shared Fixed-HP configuration, the five-task mean success rate
changes from $9.2\%$ to $14.4\%$ ($\Delta{+}{5.2}$).
T1 and T2 remain at zero success for both methods, while T3, T4, and T5 show
positive differences of $12$, $6$, and $8$ percentage points,
respectively.

We interpret this result conservatively.
It shows that the auxiliary objective can produce a positive mean difference
in a simplified navigation setting, but it does not show that all goals
benefit or that the auxiliary losses resolve failures on tasks where both
methods remain unsuccessful.
Maze2D therefore serves as an additional boundary case rather than as a
central performance claim.

\paragraph{Pixel observations provide an additional modality extension.}

The pixel experiments evaluate the same comparative objective when the agent
receives cached RGB frame stacks instead of proprioceptive state vectors.
The primary domain-level comparisons yield differences of
$+25.3$ on Walker, $-8.5$ on the complete ten-seed Cheetah comparison,
and $+5.6$ on Quadruped.
The unweighted average of the three domain means changes from $685.5$ to
$693.0$, although this aggregate is descriptive because the domains use
different numbers of matched seeds.

Cheetah is particularly sensitive to seed composition.
The complete ten-seed comparison favors TD-JEPA, whereas the fixed random
five-seed diagnostic favors TwinJEPA by $+4.3$.
We therefore do not use the pixel experiments to claim modality-independent
performance improvement.
Their main role is to show that the same training objective can be applied to
visual observations while preserving the domain- and seed-dependent behavior
seen elsewhere in the evaluation.

\paragraph{Overall interpretation.}

The extended evaluation supports a narrower conclusion than uniform
superiority across benchmarks.
TwinJEPA augments temporal predictive representation learning with two forms
of supervision derived from locally matched action comparisons.
The theoretical analysis shows what these objectives encourage on the mined
pair distribution, while the experiments show that their downstream effects
are heterogeneous.

Positive mean differences appear across the primary AntMaze evaluation and
several continuous-control settings, but neither the task-level effects nor
the component contributions are uniform.
Some tasks regress, some auxiliary variants outperform the full objective,
the magnitude of the result changes with $\lambda$, and the pixel Cheetah
comparison is sensitive to seed composition.

Taken together, these results support the claim that local reward-gap and
preference supervision can provide useful information beyond transition-wise
temporal prediction in some offline goal-conditioned control settings.
They do not establish that the supervision is universally beneficial, that
the full objective is always preferable to its individual components, or that
the theoretical properties of TD-JEPA are preserved unchanged at nonzero
auxiliary weights.
The empirical results are therefore best understood as characterizing where
and under what evaluated conditions the additional local-comparison
supervision is useful.

\paragraph{Limitations and future directions.}
The effectiveness of TwinJEPA depends on the quality and coverage of the
mined local comparison distribution.
The current implementation uses fixed pair-construction rules and
benchmark-specific neighborhood parameters, which may not identify the most
informative comparisons across all domains.
Future work will investigate more adaptive pair construction strategies,
automatic selection of informative comparison neighborhoods, and broader
evaluation across diverse offline-control and embodied-learning settings.

\end{document}